\documentclass[draftclsnofoot, onecolumn, comsoc, 12pt]{IEEEtran} 
\IEEEoverridecommandlockouts
\usepackage{graphicx}
\usepackage[group-separator={,},group-minimum-digits=4]{siunitx}
\usepackage{lipsum}
\usepackage{algorithm}
\usepackage{color, colortbl}
\definecolor{Gray}{gray}{0.9}
\definecolor{White}{gray}{1}
\usepackage{makecell} 
\usepackage{array}
\newcolumntype{P}[1]{>{\centering\arraybackslash}m{#1}}
\newcolumntype{A}{>{\columncolor{White}}P}
\usepackage{algpseudocode}

\usepackage{cite}
\usepackage{subfigure}
\usepackage{multirow} 
\usepackage{amsmath,amssymb,amsfonts}
\usepackage{float}
\usepackage{amsthm}
\usepackage{graphicx}
\usepackage{textcomp}
\usepackage{comment}
\usepackage[utf8]{inputenc}
\usepackage[english]{babel}
\usepackage{booktabs}
\usepackage[noabbrev]{cleveref}
\usepackage[short,c1,nocomma]{optidef}   
\def\BibTeX{{\rm B\kern-.05em{\sc i\kern-.025em b}\kern-.08em
    T\kern-.1667em\lower.7ex\hbox{E}\kern-.125emX}}
\newtheorem{prop}{Proposition}
\newtheorem{remark}{Remark}

\def\tcb{\textcolor{black}}
\def\tcr{\textcolor{black}}
\def\tcbb{\textcolor{black}}

\begin{document}
\bstctlcite{IEEEexample:BSTcontrol}
\title{Embracing Beam-Squint Effects for\\Wideband LEO Satellite Communications:\\A 3D Rainbow Beamforming Approach}

\author{Juha Park, Seokho Kim,  Wonjae Shin, and H. Vincent Poor
    \thanks{
    }
    \thanks{J. Park, S. Kim, and W. Shin are with the School of Electrical Engineering, Korea University, Seoul 02841, South Korea 
    (email: {\texttt{\{juha, seokho98, wjshin\}@korea.ac.kr}}); H. V. Poor is with the Department of Electrical and Computer Engineering, Princeton University, Princeton, NJ 08544, USA (email: {\texttt{poor@princeton.edu}}).
    }} 
\maketitle
\begin{abstract}
\tcb{Low Earth Orbit (LEO) satellite communications (SATCOM) offers high-throughput, low-latency global connectivity to a very large number of users. To accommodate this demand with limited hardware resources, beam hopping (BH) has emerged as a prominent approach in LEO SATCOM. However, its time-domain switching mechanism confines coverage to a small fraction of the service area during each time slot, exacerbating uplink throughput bottlenecks and latency issues as the user density increases.} Meanwhile, wideband systems experience the \emph{beam-squint effect}, where analog beamforming (BF) directions vary with subcarrier frequencies, potentially causing misalignment at certain frequencies, thereby hindering the performance of wideband SATCOM. 
In this paper, we aim to shift the paradigm in wideband LEO SATCOM from beam-squint as an impairment to beam-squint as an asset. Specifically, we put forth \emph{3D rainbow BF} employing a joint phase-time array (JPTA) antenna with true time delay (TTD) to intentionally widen the beam-squint angle, steering frequency-dependent beams toward distributed directions. This novel approach enables the satellite to serve its entire coverage area in a single time slot. By doing so, the satellite simultaneously receives uplink signals from a massive number of users, significantly boosting throughput and reducing latency. To realize 3D rainbow BF, we formulate a JPTA beamformer optimization problem and address the non-convex nature of the optimization problem through a novel joint alternating and decomposition-based optimization framework. \tcb{Through numerical evaluations incorporating realistic 3D LEO SATCOM geometry, our numerical results demonstrate that the proposed rainbow BF-empowered LEO SATCOM achieves up to \num{2.8}-fold increase in uplink throughput compared to conventional BH systems.} These results mark a significant breakthrough for wideband LEO SATCOM, paving the way for high-throughput, low-latency global connectivity.
\end{abstract}


\section{Introduction}
Low Earth orbit (LEO) satellite communications (SATCOM) has recently garnered considerable attention from academia and industry \cite{kodheli2020satellite}. Unlike terrestrial networks, SATCOM provides ubiquitous connectivity for global service delivery across diverse geographical regions. Furthermore, operating at relatively low altitudes ($\num{300}$-$\num{2000}$ km) compared to geostationary Earth orbit (GEO) and medium Earth orbit (MEO) satellites, LEO satellites achieve reduced latency and enhanced data rates. However, due to the relatively narrow per-satellite coverage of LEO satellites compared to MEO and GEO satellites, LEO constellations with massive satellites are required to provide global connectivity. For example, Starlink's constellation is expected to consist of approximately \num{42000} satellites across different orbits \cite{susanto2024analysis}.
Consequently, reducing the implementation costs of the LEO mega-constellation is a fundamental requirement. Fortunately, the recent development of reusable rockets has significantly reduced LEO satellite launch costs. However, research on cost-effective hardware and signal processing designs remains constrained by conventional narrowband assumption-based design paradigms. This limitation impedes innovative technological advancement in LEO SATCOM systems.

In LEO SATCOM systems, one of the primary challenges in ensuring reliable communication links is significant path loss, resulting from the long propagation distances between satellites and ground user terminals. For instance, at a frequency of 14 GHz and a satellite-user distance of \num{500} km, the free-space path loss amounts to \num{169.35} dB. To overcome this substantial path loss, high-gain beamforming (BF) is essential at both satellite and user terminals. Analog BF technologies, particularly phased array (PA) antenna systems, have emerged as a practical and cost-effective solution, enabling highly directional beams while requiring only a few radio-frequency (RF) chains. For example, Starlink has used a phased array with \num{1280} antenna elements in its user terminal \cite{yang2023starlink}. However, three critical technical challenges exist in analog BF-based SATCOM systems as follows: \emph{i) wideband beam-squint effect}, \emph{ii) limited number of simultaneously active beams}, and \emph{iii) uplink throughput bottleneck in beam hopping (BH) systems}.

\subsection{Technical Challenges}
{\it i)} \textbf{Wideband beam-squint effect}: Although state-of-the-art LEO SATCOM systems (e.g., Starlink) successfully provide global connectivity, their throughput remains inferior to that of terrestrial networks such as 5G new radio (NR). Consequently, LEO SATCOM systems face significant challenges in attracting urban users who can readily access conventional terrestrial networks. Moreover, considering next-generation services that require high data rates, such as digital twins, virtual reality (VR), and extended reality (XR), LEO SATCOM must evolve to support wider bandwidth. Conventional phase shifter (PS)-based analog BF works well in narrowband systems, where subcarrier wavelengths are nearly identical, but faces significant challenges in wideband systems. In wideband analog BF, the distinct wavelengths of subcarriers induce the beam-squint effect \cite{wang2019beam,wan2021hybrid,you2022beam,dai2022delay}, which causes frequency-dependent deviations in beam direction, leading to misalignment between the transmitted beam and the intended target, as shown in Fig. \ref{[Introduction]Beam-squint effect}.
\begin{figure}[t]
\centering
\includegraphics[width=0.85\linewidth]{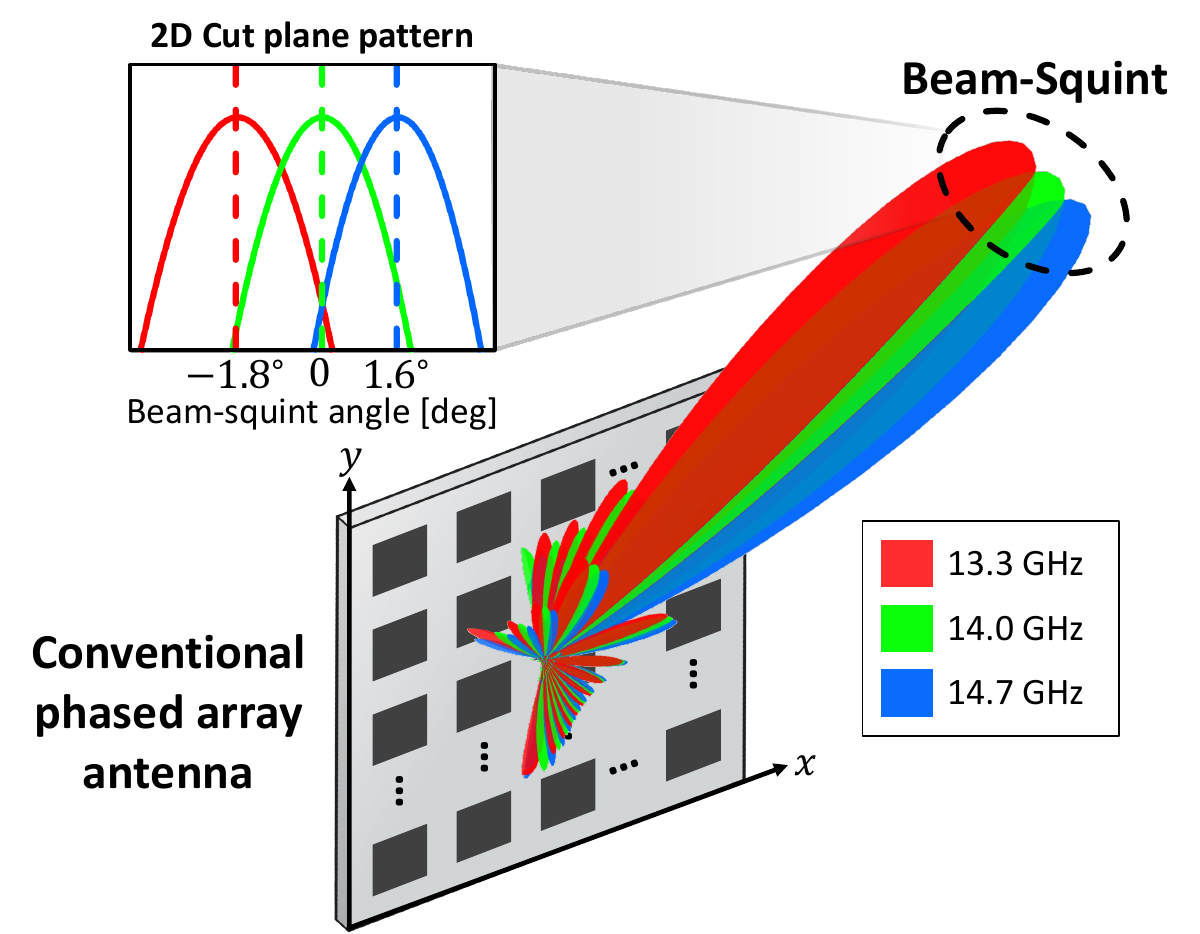}
\caption{Beam-squint effect with $16\times16$ uniform rectangular array antenna with a center frequency of $14$ GHz and a bandwidth of $1.4$ GHz.}
\label{[Introduction]Beam-squint effect}
\end{figure}
The misalignment results in reduced BF gain, significantly degrading system performance. Therefore, the beam-squint effect must be carefully addressed in the design of wideband LEO SATCOM to ensure reliable and efficient communication.

{\it ii)} \tcb{\textbf{Limited number of simultaneously active beams}: Due to their high altitude, LEO satellites provide much broader coverage than terrestrial base stations.} \tcb{In analog BF systems, the number of simultaneously active beams is typically limited by the number of RF chains \cite{ratnam2022joint}.} Each RF chain, comprising various analog/digital signal processing components, e.g., mixers, amplifiers, and converters, demands substantial hardware costs and power consumption. Consequently, increasing the number of RF chains to accommodate additional simultaneous beams becomes impractical in LEO SATCOM systems where hardware cost and energy efficiency are critical design constraints. To serve massive users under these constraints, BH has emerged as a viable solution \cite{DVB,3GPP_BH,tang2021resource,zhang2023system,li2021overview,ouyang2025dependency,lin2022dynamic,8984132,tang2021optimization}. BH operates through quasi-periodic beam-switching in predefined time-space transmission patterns, also known as BH time plan (BHTP). By leveraging dynamic beam-switching, BH enables satellites to serve spatially distributed massive users with limited payload capacity of the satellite. \tcb{BH} has garnered significant attention not only in academia but also from industry and standardization organizations. Specifically, the DVB-S2X standard introduced three frame structures (formats 5-7) with variable frame length and dummy frame capabilities to support the implementation of BH \cite{DVB}. The $3^{\rm rd}$ generation partnership project (3GPP) radio access network working group 1 (RAN WG1) is considering BH technology to enhance downlink coverage in NR-based non-terrestrial network (NR-NTN) systems, as part of discussions in Release 19 \cite{3GPP_BH}. \tcb{While BH is a promising technology for overcoming the limitation of simultaneous active beams in analog BF-based LEO SATCOM, it introduces additional critical challenges arising from its time-domain beam-switching nature.}

{\it iii)} \textbf{Uplink throughput bottleneck in BH systems:}
The dynamic beam-switching mechanism of BH leads to an uplink throughput bottleneck in LEO SATCOM systems by limiting the full utilization of available power resources in the uplink scenarios. To clarify this fundamental limitation, it is essential to examine the asymmetric power utilization characteristics between downlink and uplink scenarios of BH systems. In particular, consider $K$ users distributed within the satellite's service coverage, where only a subset of $K'<K$ users are served by the satellite's active beams at any given time instant in BH systems. In downlink transmission, the satellite operates with a centralized power budget $P_{\sf sat}$ and can achieve full power utilization regardless of the instantaneous number of served users $K'$. Conversely, in uplink transmission, each user possesses an individual power budget $P_{\sf user}$, yet only $K'$ users can transmit concurrently. \tcbb{Although users located outside the beam footprints may still transmit, their signals experience severe attenuation due to insufficient beamforming gain, resulting in negligible contribution to the overall system performance. Consequently, it is more efficient to allocate uplink resources exclusively to user located within the beam footprints. As a result, the total utilized uplink power is limited to $K' P_{\sf user}$, while a substantial portion $(K-K')P_{\sf user}$ remains unutilized due to the exclusion of $(K-K')$ users from the beams. Moreover, the power budget of user terminals is more strictly constrained than that of satellites, primarily due to RF exposure regulations and the limited operating ranges of RF amplifiers (e.g., $P_{\sf ut}=33$ dBm for very small aperture terminal (VSAT) and $P_{\sf ut}=23$ dBm for handheld terminal \cite{3gpp38.821}). This limited utilization of the power budget causes a significant uplink throughput bottleneck in LEO SATCOM systems. The problem becomes particularly severe due to the characteristics of LEO SATCOM, where a single satellite must serve a massive number of users simultaneously. In addition, the periodic nature of BH also introduces inherent latency issues.} The uplink
throughput bottleneck in the BH system remains one of the most critical challenges impeding the advancement of LEO SATCOM. Despite its importance, this issue has yet to receive adequate attention in existing literature.

\subsection{Related Works}
\tcb{\tcr{Previous} research investigating the beam-squint effect has primarily focused on terrestrial networks, e.g., performance analysis \cite{yu2021performance,noh2021design}, channel estimation \cite{wang2019beam,jian2019angle,xu2024overcoming}, and precoding/BF \cite{gao2021wideband,chen2020hybrid,you2022beam,ratnam2022joint,alammouri2022extending, ma2023beam, ma2025switch}. Specifically, \cite{yu2021performance} demonstrated that traditional BF codebooks based on narrowband assumptions suffer significant performance degradation when dealing with wide bandwidths. A closed-form metric to evaluate beam-squint, referred to as the ``beam-squint ratio," \tcr{was} proposed in \cite{ma2023beam}. Furthermore, it \tcr{was} shown that the beam squint ratio increases linearly with fractional bandwidth. In \cite{wang2019beam}, a compressed sensing-based channel estimation technique was proposed to mitigate the beam-squint effect in wideband massive multiple-input multiple-output (MIMO) systems. Furthermore, \cite{you2022beam} investigated LEO satellite-based integrated sensing and communications (ISAC) and proposed beam-squint-aware BF optimization techniques. A cost-effective beam-squint mitigation technique based on switch-based hybrid BF architecture \tcr{was} proposed in \cite{ma2025switch}.} \tcb{Recent studies \cite{ratnam2022joint,alammouri2022extending,gao2023integrated,luo2024yolo,kim2023fast,zhai2021ss} have highlighted the potential benefits of leveraging \tcr{the} beam-squint effect in terrestrial systems through controlled beam-squint using the frequency-dependent phase shifts of true time delay (TTD).} Specifically, \cite{luo2024yolo} proposed a one-shot sensing method using a rainbow beam, which intentionally widens the beam-squint angle to achieve simultaneous coverage across the entire angular space. Additionally, \cite{kim2023fast} proposed a rainbow BF-enabled fast beam management method for terrestrial networks. However, the exploitation and potential advantages of the beam-squint effect, such as rainbow BF, have been unexplored in the context of LEO SATCOM. Recently, \cite{sekimori2024frequency} proposed a rainbow BF technique (referred to as ``frequency prism" in their work) for LEO SATCOM utilizing delay-adjustable intelligent reflecting surfaces. Built on the simplified 2D structure, the authors presented a frequency utilization efficiency maximization approach to accommodate heterogeneous throughput demands. However, their framework is not directly applicable to real-world 3D LEO SATCOM. \tcb{Notably, extending 2D rainbow BF to 3D is non-trivial, as 3D rainbow BF must account for both azimuth and elevation (or off-nadir) angles, significantly increasing the degrees of freedom in BF design. In 2D rainbow BF, different beams can be collectively steered across the full angular coverage by simply widening the beam-squint angle in an angular sector, which has a closed-form solution \cite{kim2023fast}. However, in 3D rainbow BF, widening the beam-squint cannot fully cover the entire service coverage with frequency-dependent beams. Hence, unlike 2D rainbow BF where controlling the beam-squint angle is sufficient, 3D rainbow BF requires an additional frequency-direction mapping to precisely determine the 3D beam direction, specifying which frequency components should be directed toward specific spatial directions while incorporating both azimuth and elevation angles. Furthermore, developing algorithms to achieve such mapping presents significantly greater challenges compared to 2D rainbow BF.}

BH has attracted considerable research interest in recent years. In \cite{tang2021resource}, the authors developed resource allocation strategies when LEO satellites share a GEO satellite's spectrum in a BH manner. System-level evaluations of BH in NR-based LEO SATCOM systems were conducted in \cite{zhang2023system}. Furthermore, \cite{li2021overview} analyzed BH algorithms for large-scale LEO constellations, particularly focusing on optimization techniques for BHTP that maximize coverage efficiency. A cooperative multi-agent deep reinforcement learning (MADRL) approach for jointly optimizing BHTP and bandwidth allocation in response to dynamic traffic demands was introduced in \cite{lin2022dynamic}. The authors in \cite{ouyang2025dependency} proposed a dependency-elimination MADRL framework that integrates feeder- and user-links for jointly allocating beams, power, and bandwidth while reducing computational complexity. One of the key technical challenges in BH systems is their inherent latency. Data packets awaiting transmission must queue until their corresponding beam is illuminated \cite{8984132}. Therefore, when designing the BHTP, it is crucial to ensure that the beam revisit time remains below a specified threshold to maintain reliable real-time services. Exceeding this threshold can significantly degrade the quality of experience. Addressing this latency concern, \cite{tang2021optimization} proposed a cell division optimization method for LEO BH satellite communication systems to minimize packet queueing delay. The aforementioned studies have established BH as a viable solution to provide broad coverage with limited hardware resources. However, existing literature has not addressed the uplink throughput bottlenecks stemming from the fundamental limitations of BH's time-domain beam-switching mechanisms. Furthermore, existing BH research has primarily focused on narrowband systems, limiting its extension to wideband systems. In wideband systems, the beam-squint effect induces frequency-dependent variations in beam directionality, potentially degrading the performance of conventional BH systems that rely on the narrowband assumption.

\subsection{Motivations and Contributions}
\tcbb{In this paper, we aim to overcome the fundamental origin of the uplink throughput bottleneck in conventional BH-based LEO SATCOM by \emph{embracing} the beam-squint effects rather than mitigating them as in existing literature \cite{gao2021wideband,chen2020hybrid,you2022beam, ma2023beam, ma2025switch}. To achieve this goal, we propose algorithms for both analog and digital domain optimization. First, we put forth a joint phase-time array (JPTA) based analog BF optimization algorithm that exploits beam squint to maximize the number of simultaneously serviceable users. Specifically, the proposed approach intentionally steers the frequency-dependent beams at each subcarrier toward spatially distributed directions across the coverage area, referred to as \emph{3D rainbow BF}. This design enables simultaneous service of the entire coverage region within a single time slot using only a single RF chain. By doing so, the limited uplink power utilization, which is the origin of the uplink throughput bottleneck, can be eliminated. Next, we propose a joint subcarrier and power allocation (JSPA) algorithm for the proposed rainbow BF-empowered LEO SATCOM system with extremely short channel coherence time. The JSPA algorithm leverages statistical and geometric channel state information (CSI) to optimize resource allocation. The key contributions of this work are summarized as follows:}
\begin{itemize}
   \item We propose a 3D rainbow BF design framework to overcome the uplink throughput bottleneck in LEO SATCOM. First, we establish the conditions for achieving full BF gain across all subcarriers while ensuring accurate beam steering based on a given frequency-direction mapping, which determines how each subcarrier frequency's beam is aligned with its intended direction. We then mathematically prove that, in general, no combination of TTD and PS values can fully satisfy these conditions. Therefore, to achieve the best possible frequency-direction mapping, we formulate an optimization problem that minimizes the deviation between the JPTA beamformer and the desired rainbow beamformer.
   \item To solve the non-convex rainbow BF optimization problem with coupled antenna elements, we develop a joint alternating and decomposition-based optimization algorithm. \tcb{First, we derive closed-form optimal phase rotation coefficients for given TTD and PS values, considering the interdependence of phase rotation across all antenna elements. Second, we derive closed-form optimal PS values for given TTD and phase rotation coefficients. Third, we decompose the original high-dimensional problem into antenna element-wise 1D line search problems, which we solve to efficiently determine the optimal TTD values for given phase rotation coefficients. Finally, by iteratively updating the optimization variables, our algorithm ensures convergence to near-optimal solutions.}
   
   \item \tcb{To address the challenge of acquiring instantaneous CSI due to the extremely short channel coherence time in LEO SATCOM, we propose a JSPA algorithm leveraging statistical and geometric CSI.} We formulate the throughput maximization problem using the approximated rate expression. To address the optimization problem, we derive the necessary optimality conditions for subcarrier and power allocation by mathematically manipulating Karush-Kuhn-Tucker (KKT) conditions. The intertwined optimality conditions for subcarrier and power allocation are efficiently handled through a greedy approach, ensuring robust satellite resource management.
   
   \item Through numerical evaluations incorporating realistic 3D LEO SATCOM geometries, we demonstrate that the proposed rainbow BF system outperforms baseline schemes, specifically BH and beam sharing \cite{he2022physical}. \tcb{Compared to conventional BH, the widely-used and state-of-the-art scheme, our proposed rainbow BF system demonstrates the capability to simultaneously serve up to \num{12.1} times more users and \num{2.8} times higher uplink throughput. Notably, the performance gap between rainbow BF and conventional BH widens with increasing user density and bandwidth, showcasing the \emph{scalability} of rainbow BF. In addition, we show that the proposed JSPA algorithm achieves near-optimal performance. Moreover, in LOS-dominant scenarios, the performance is close to the ideal case with perfect channel information, despite only relying on partial CSI.}
\end{itemize}

\subsection{Notations}
Herein, standard letters, lower-case boldface letters, and upper-case boldface letters indicate scalars, vectors, and matrices, respectively. The imaginary unit is defined as $j\triangleq\sqrt{-1}$. Notations $(\cdot)^{\sf{T}}$, $(\cdot)^{\sf{H}}$, $\Vert\cdot\Vert$ identify the transpose, conjugate transpose, and $\ell_2$-norm, respectively. $\mathrm{vec}(\cdot)$ denotes a column-wise vectorization operator. $\mathbb{E}[\cdot]$ indicates the statistical expectation. $\mathcal{CN}(\mu,\sigma^2)$ \tcb{represents} a complex Gaussian random variable with a mean of $\mu$ and a variance $\sigma^2$. $\angle(\cdot)$ is phase of complex number; $[x]^{+}\triangleq \max\{x,0\}$ for any $x\in\mathbb{R}$. $\otimes$ denotes the Kronecker product operator. $\Re(\cdot)$ denotes the real part of a complex number.
\section{System Model}
\begin{figure}[!t]
\centering
\includegraphics[width=1\linewidth]{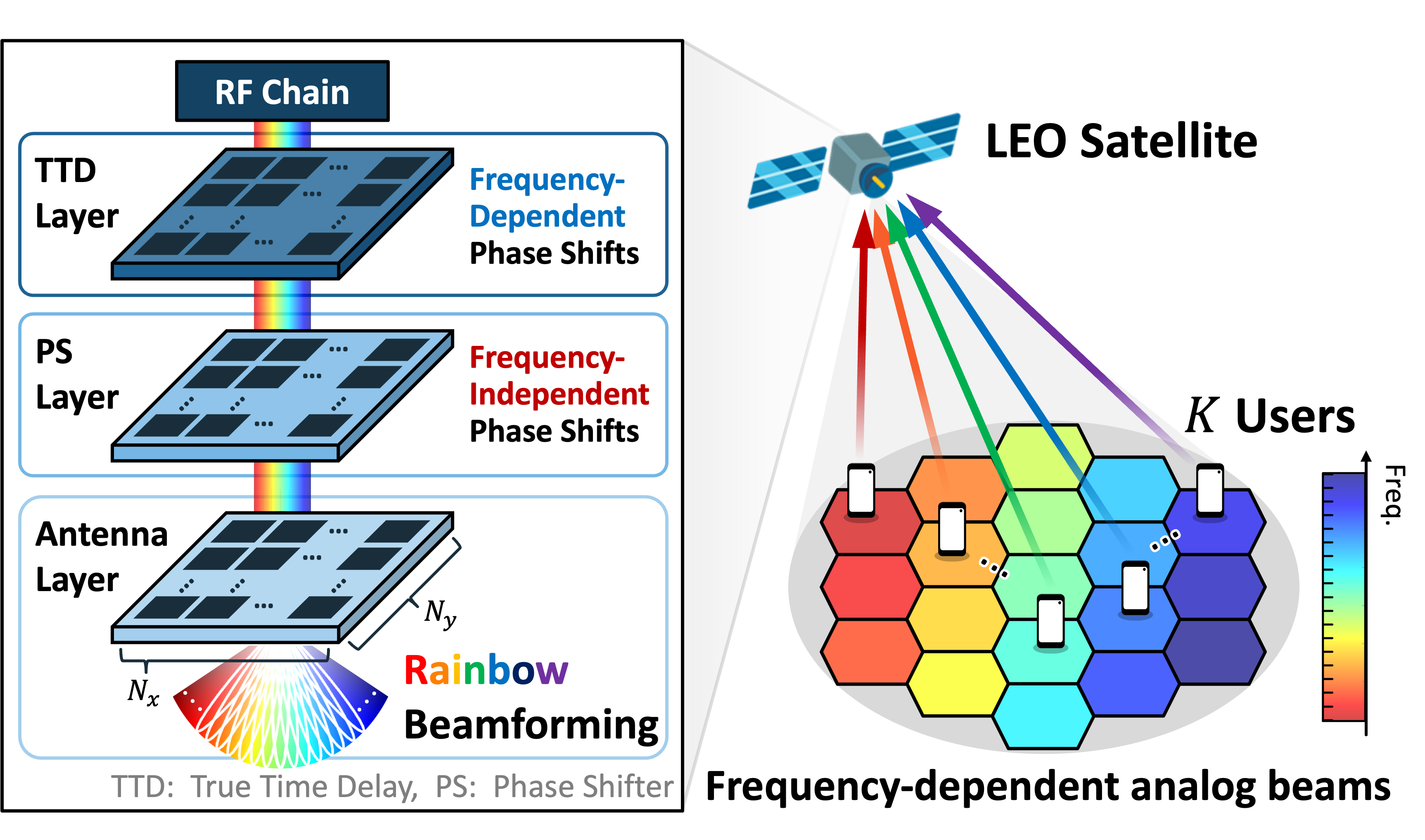}
\caption{System model of the proposed rainbow BF-empowered LEO SATCOM systems.}
\label{block_diagram}
\end{figure}
We consider an uplink orthogonal frequency division multiple access (OFDMA) LEO SATCOM system. Each subcarrier is exclusively allocated to a single user. The center frequency is denoted as $f_{\sf c}$, and the $m$-th subcarrier frequency is $f_m\triangleq f_{\sf c}+(m-\frac{M+1}{2})\Delta f$ for all $m\in\{1,\cdots,M\}$, where $\Delta f$ is the subcarrier spacing. The satellite employs a JPTA antenna with uniform rectangular array (URA) geometry comprising $N_{\sf rx}=N_x\times N_y$ elements, where $N_x$ and $N_y$ denote the number of antenna elements in the $x$ and $y$ directions, respectively, as illustrated in Fig. 2. Each antenna element is individually connected to a dedicated TTD and PS in a one-to-one manner. The satellite operates with a single RF chain. \tcb{On the Earth's surface, $K$ single-antenna users are distributed across the satellite's wide coverage area, each constrained by an instantaneous power budget of $P_k$ due to limited power amplifier's operating range and RF exposure regulations.}
\subsection{3D Angular Representation in UV-plane}
\begin{figure}[!t]
\centering
\includegraphics[width=0.85\linewidth]{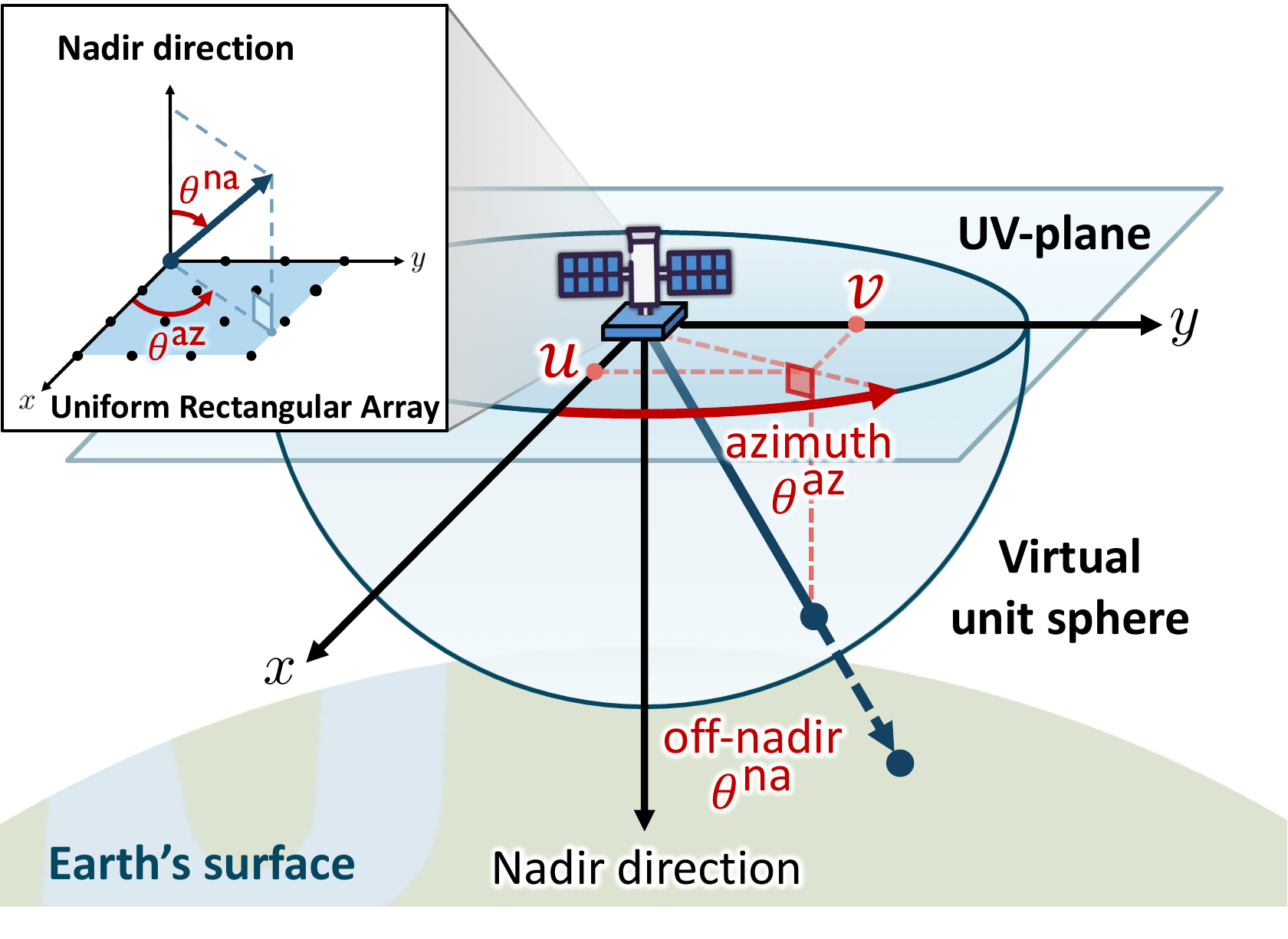}
\caption{Visualization of the relationship between azimuth/off-nadir angles and their UV-plane angular representations.}
\label{Array geometry}
\end{figure}
In this work, we employ UV-plane representations rather than spherical coordinates to design and analyze 3D BF more efficiently. As shown in Fig. \ref{Array geometry}, the UV-plane representation maps spherical coordinate angular representation (off-nadir angle $\theta^{\sf na}$ and azimuth angle $\theta^{\sf az}$) to virtual coordinates $(u,v)$. The mapping is defined as $u\triangleq\sin\theta^{\sf na}\cos\theta^{\sf az}$ and $v\triangleq\sin\theta^{\sf na}\sin\theta^{\sf az}$. This spherical-to-UV mapping can be interpreted as the projection onto the UV-plane of the intersection point between a line (determined by azimuth and off-nadir angles) and a virtual unit sphere centered at the satellite. Thus, a one-to-one correspondence holds between points on the Earth's surface and UV-plane coordinates $(u,v)$ as shown in Fig. \ref{Array geometry}. \tcb{Note that $u\in[-1,1]$, $v\in[-1,1]$, and $u^2+v^2\leq1$ must be satisfied by definition.}
\subsection{Uplink Channel Model}
\tcb{We adopt the widely used LEO SATCOM channel model \cite{you2020massive,you2022beam,10680149}. We assume that the subcarrier spacing is sufficiently narrow to ensure flat fading for each subcarrier and that Doppler shift can be compensated. Defining $(u_k,v_k)$ as the UV-plane representation of the angle-of-arrival (AOA) for the $k$-th user and assuming that AOA of each multipath are approximately the same \cite{you2020massive,you2022beam,10680149}, the channel between the $k$-th user and the satellite at the $m$-th subcarrier is given by}
\begin{align}
&\mathbf{h}_k^{(m)}=g_k^{(m)}\mathbf{a}^{(m)}(u_k,v_k).
\end{align}
Here, $g_k^{(m)}\in\mathbb{C}$ is a complex channel gain; $\mathbf{a}^{(m)}(u_k,v_k)$ is a frequency-dependent array response vector, given by
 \begin{align}
 \mathbf{a}^{(m)}(u_k,v_k) = \mathbf{a}_x^{(m)}(u_k)\otimes\mathbf{a}^{(m)}_y(v_k),
 \end{align}
 where $\mathbf{a}_x^{(m)}(u_k)$ and $\mathbf{a}_y^{(m)}(v_k)$ are
\begin{align}
&\mathbf{a}^{(m)}_x\left(u_k\right)=\begin{bmatrix}
1,e^{-j\pi\frac{f_m}{f_{\sf c}}u_k},
\cdots,  e^{-j(N_x-1)\pi\frac{f_m}{f_{\sf c}}u_k}
\end{bmatrix}^{\sf T},
\nonumber
\\
&\mathbf{a}_y^{(m)}\left(v_k\right)=\begin{bmatrix}
1,e^{-j\pi\frac{f_m}{f_{\sf c}}v_k},
\cdots,  e^{-j(N_y-1)\pi\frac{f_m}{f_{\sf c}}v_k}
\end{bmatrix}^{\sf T}.
\end{align} \tcb{This frequency-dependent array response vector causes beam-squint effects when using conventional PA with PSs. Since the phase shift provided by PS is frequency-independent, PS cannot fully compensate for the frequency-dependent phase shifts of the array response $\mathbf{a}^{(m)}(u_k,v_k)$. This mismatch manifests as the beam-squint effects, where beams at different frequency components point toward different spatial directions.} We model $g_k^{(m)}$ as a complex random variable following a Rician distribution, assuming the existence of a LOS path. Specifically, the real and imaginary parts of $g_k^{(m)}$ are independently and identically distributed, following $\mathcal{N}\left(\sqrt{\frac{\kappa_k\eta_k^{(m)}}{2(\kappa_k+1)}},\frac{\eta_k^{(m)}}{2(\kappa_k+1)}\right)$ \cite{you2020massive}. Here, $\kappa_k$ and $\eta_k^{(m)}$ represent the Rician factor and the average channel power, i.e., $\eta_k^{(m)}\triangleq\mathbb{E}\big[|g_k^{(m)}|^2\big]$.

\subsection{Uplink Signal Model and Throughput}
Our signal model incorporates multiple time slots to analyze the time-averaged performance. We consider a total of $L$ time slots, where $\ell\in\{1,\ldots,L\}$ denotes the time slot index. In the baseband model, the transmit signal of the $k$-th user on the $m$-th subcarrier during the $\ell$-th time slot is expressed as
\begin{equation}
x_k^{(m,\ell)}=b_k^{(m,\ell)}\sqrt{p_k^{(m,\ell)}}s_k^{(m,\ell)}.
\end{equation}
Here, $b_k^{(m,\ell)}\in\{0,1\}$ is the binary subcarrier allocation index, which is 1 if the subcarrier $m$ is allocated to the user $k$ at time slot $\ell$, and 0 otherwise; $p_k^{(m,\ell)}$ is the power allocation factor; $s_k^{(m,\ell)}$ is the unit-power symbol, i.e., $\mathbb{E}[|s_k^{(m,\ell)}|^2]=1$. 

At the transmitter, i.e., ground user, \tcb{the baseband symbols} $\big\{x_k^{(1,\ell)},x_k^{(2,\ell)},\cdots,x_k^{(M,\ell)}\big\}$ are multiplexed and converted to a passband OFDM signal $\Re\big\{\sum_{m=1}^M x_k^{(m,\ell)}e^{j2\pi f_mt}\big\}$. Consequently, this signal propagates through the channel and passes through the satellite's JPTA beamformer. Let $\Re\big\{\sum_{m=1}^M r_k^{(m,\ell)}e^{j2\pi f_mt}\big\}$ denote the received passband signal at the $(n_x,n_y)$-th satellite antenna element, where $r_k^{(m,\ell)}$ incorporates channel fading effect. This signal first passes through the $(n_x,n_y)$-th PS, where its phase is shifted by $\phi^{(n_x,n_y,\ell)}$, resulting in 
\begin{align}
\Re\left\{\sum_{m=1}^M r_k^{(m,\ell)}e^{j\phi^{(n_x,n_y,\ell)}}e^{j2\pi{f_m}t}\right\}.
\end{align}
Subsequently, this signal passes through the $(n_x,n_y)$-th TTD, introducing a delay of $\tau^{(n_x,n_y,\ell)}$, which yields 
\begin{align}
&\Re\left\{\sum_{m=1}^M r_k^{(m,\ell)}e^{j\phi^{(n_x,n_y,\ell)}}e^{j2\pi{f_m}\left(t-\tau^{(n_x,n_y,\ell)}\right)}\right\}
\nonumber\\
&=\Re\left\{\sum_{m=1}^M r_k^{(m,\ell)}e^{j\left\{\phi^{(n_x,n_y,\ell)}-2\pi f_m\tau^{(n_x,n_y,\ell)}\right\}}e^{j2\pi f_m t}\right\}.
\end{align}
After analog-to-digital sampling and fast Fourier transform (FFT) operation, the received baseband symbol for subcarrier $m$ is $
r_k^{(m,\ell)}e^{j\left\{\phi^{(n_x,n_y,\ell)}-2\pi f_m\tau^{(n_x,n_y,\ell)}\right\}}$. \tcb{Note that although PS and TTD values $\{\phi^{(n_x,n_y,\ell)},\tau^{(n_x,n_y,\ell)}\}$ cannot be independently configured for each subcarrier frequency $f_m$, their combined effect inherently varies across frequencies, resulting in a frequency-dependent phase shift. While TTD alone yields a frequency-dependent phase shift, adding PS introduces an additional degree of freedom, enabling a more flexible frequency-dependent BF control.}

\tcb{We define the frequency-dependent JPTA BF matrix as $\mathbf{W}^{(m)}(\mathbf{T}^{(\ell)},\boldsymbol{\Phi}^{(\ell)})\in\mathbb{C}^{N_x\times N_y}$, which is a function of \tcr{the} time delay matrix $\mathbf{T}^{(\ell)}$ and phase shift matrix $\boldsymbol{\Phi}^{(\ell)}$ at time slot $\ell$. Here, the $(n_x,n_y)$-th elements of $\mathbf{T}^{(\ell)}\in\mathbb{R}^{N_x \times N_y}$ and $\boldsymbol{\Phi}^{(\ell)}$ are $\tau^{(n_x,n_y,\ell)}$ and $\phi^{(n_x,n_y,\ell)}$, which represent the delay of the TTD and the phase shift of the PS corresponding to the $(n_x,n_y)$-th element in the $\ell$-th time slot, respectively. Consequently, the $(n_x,n_y)$-th element of $\mathbf{W}^{(m)}(\mathbf{T}^{(\ell)},\boldsymbol{\Phi}^{(\ell)})$ represents the phase shift corresponding to the $(n_x,n_y)$-th antenna element, given by $e^{j\left\{\phi^{(n_x,n_y,\ell)}-2\pi f_m\tau^{(n_x,n_y,\ell)}\right\}}$. For notational brevity, we define \tcr{a} vectorized version of \tcr{the} JPTA BF matrix as $\mathbf{w}^{(m)}\big(\mathbf{T}^{(\ell)},\boldsymbol{\Phi}^{(\ell)}\big)\triangleq{\sf vec}\big(\mathbf{W}^{(m)}(\mathbf{T}^{(\ell)},\boldsymbol{\Phi}^{(\ell)}\big)^{\sf T}\big)$.} Then, the satellite's received signal after BF can be expressed as
\begin{equation}
y^{(m,\ell)}=\mathbf{w}^{(m)}\big(\mathbf{T}^{(\ell)},\boldsymbol{\Phi}^{(\ell)}\big)^{\sf H}\left\{\sum_{k=1}^K \mathbf{h}_k^{(m)} x_k^{(m,\ell)}+\mathbf{z}^{(m,\ell)}\right\},
\end{equation} 
where $\mathbf{z}^{(m,\ell)} \sim \mathcal{CN}(\mathbf{0},\sigma^2\mathbf{I}_{N_{\sf rx}})$ is complex Gaussian noise. Here, we assume exclusive subcarrier allocation, i.e., $b_k^{(m,\ell)}\in\{0,1\}$ and $\sum_{k=1}^K b_k^{(m,\ell)}=1$; thus, there is no multi-user interference.
\tcb{By assuming that the channel remains invariant during $L$ time slots, the time-averaged sum throughput can be expressed as follows:}
\begin{align}
&\frac{1}{L}\sum_{\ell=1}^L\sum_{m=1}^M\sum_{k=1}^K b_k^{(m,\ell)}\Delta f
\nonumber\\
&\quad\quad\quad\times\mathrm{log}_2\left(1+\frac{p_k^{(m,\ell)}\Big|\mathbf{w}^{(m)}\big(\mathbf{T}^{(\ell)},\boldsymbol{\Phi}^{(\ell)}\big)^{\sf H}\mathbf{h}_k^{(m)}\Big|^2}{N_{\sf rx}\sigma^2}\right).\label{Throughput model}
\end{align}
\tcb{In the following sections, we present algorithms for rainbow BF design and resource allocation optimization. Specifically, Sec. III proposes a rainbow beamformer design algorithm that maximizes the number of simultaneously served users, thereby increasing the total active uplink power budget. Subsequently, Sec. IV develops subcarrier and power allocation algorithms based on the pre-designed rainbow beamformer, taking into account the challenge of obtaining instantaneous CSI. As a result, the time-averaged sum throughput in Eq. \eqref{Throughput model} can be significantly improved compared to conventional BH systems. Throughout Sec. III and IV, we focus on a specific time slot and thus omit the time slot index $\ell$.}
\section{3D Rainbow Beamforming Design}
\begin{figure}[t]
\centering
\includegraphics[width=1\linewidth]{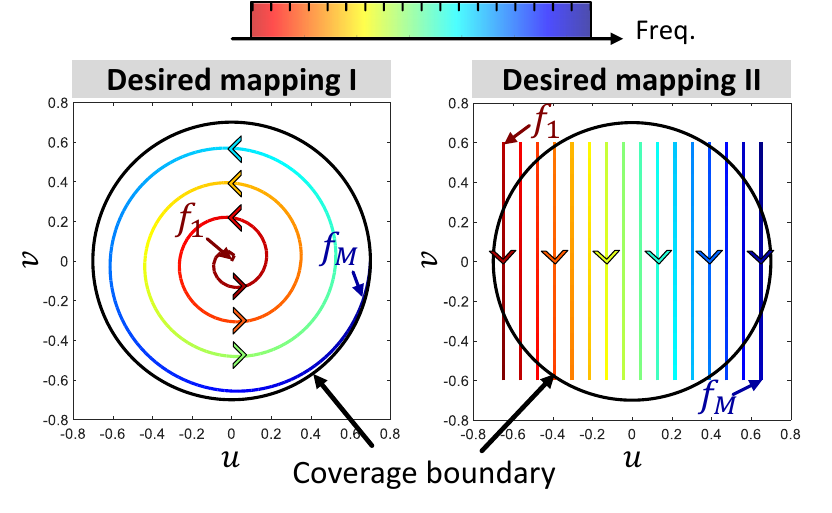}
\caption{Examples of $\{\mathbf{u}_{\sf des}, \mathbf{v}_{\sf des}\}$ for frequency-direction mapping, distributing beams at different subcarrier frequencies across the coverage area.}
\label{Rainbow BF patterns}
\end{figure}
In this section, we propose a design and optimization framework for a 3D rainbow BF for our system. The goal is to design a 3D rainbow BF that maximizes the number of simultaneously serviceable users by generating multiple frequency-dependent beams that collectively span the entire service coverage. A key aspect of this design lies in frequency-direction mapping, which determines how each subcarrier frequency's beam is aligned with its intended direction to ensure coverage and user connectivity. Specifically, we define the desired frequency-direction mapping in the vector form as $\mathbf{u}_{\sf des}\triangleq\big[u^{(1)}_{\sf des},\cdots,u^{(M)}_{\sf des}\big]^{\sf T}$ and $\mathbf{v}_{\sf des}\triangleq\big[v^{(1)}_{\sf des},\cdots,v^{(M)}_{\sf des}\big]^{\sf T}$, where $\big(u_{\sf des}^{(m)},v_{\sf des}^{(m)}\big)$ denotes the desired beam direction of subcarrier $m$. In the following, we refer to $\{\mathbf{u}_{\sf des},\mathbf{v}_{\sf des}\}$ as the desired frequency-direction mapping. We first need to design $\{\mathbf{u}_{\sf des},\mathbf{v}_{\sf des}\}$ that is capable of achieving our goal. As a representative example in Fig. \ref{Rainbow BF patterns}, desired mapping I fills the coverage area in a spiral shape as the subcarrier frequency increases, while desired mapping II fills the coverage area with multiple lines. While numerous mapping strategies exist, we first present the optimization framework assuming a pre-designed $\{\mathbf{u}_{\sf des},\mathbf{v}_{\sf des}\}$. The design criteria for these mappings will be discussed in Section III. C. 

\tcb{The 3D rainbow BF can be achieved by optimizing the combination of the frequency-flat phase shifts of PSs and the frequency-linear phase shifts of TTDs, i.e., $\mathbf{T}$ and $\boldsymbol{\Phi}$.} \tcb{A fundamental consideration in this context is the feasibility of achieving the desired frequency-direction mapping while preserving \tcr{the} BF gain.} The \tcb{BF} gain at subcarrier $m$ toward direction $(u,v)$ is expressed as $\big|\mathbf{w}^{(m)}(\mathbf{T},\boldsymbol{\Phi})^{\sf H}\mathbf{a}^{(m)}(u,v)\big|^2$. To achieve maximum BF gain in the desired direction across all subcarriers, the following conditions must be satisfied:
\begin{align}
\mathbf{w}^{(m)}(\mathbf{T},\boldsymbol{\Phi})=\alpha^{(m)}\mathbf{a}^{(m)}\left(u_{\sf des}^{(m)},v_{\sf des}^{(m)}\right),\forall m.
\label{condition_Fullgain}
\end{align}
Here, $\alpha^{(m)}$ is an arbitrary unit-modulus phase rotation coefficient, i.e., $|\alpha^{(m)}|=1$, which does not affect the beam direction and the beam gain.
In the following proposition, we prove it is impossible to achieve full BF gain across all subcarriers while simultaneously steering them following arbitrary desired frequency-direction mapping using a JPTA beamformer when the number of subcarriers is three or more. 
\begin{prop}
{\it For an arbitrary desired frequency-direction mapping $\{\mathbf{u}_{\sf des},\mathbf{v}_{\sf des}\}$, a solution $\{\mathbf{T}^\star,\boldsymbol{\Phi}^\star\}$ that satisfies Eq. \eqref{condition_Fullgain} generally does not exist when the number of subcarriers is three or more ($M\geq3$).}
\end{prop}

\begin{proof}
\tcb{We prove by contradiction. Assume a solution $\{\mathbf{T}^\star,\boldsymbol{\Phi}^\star\}$ exists satisfying Eq. \eqref{condition_Fullgain} for all subcarriers and arbitrary desired frequency-direction mapping. For three distinct subcarriers among $\{1,\cdots,M\}$, with mathematical manipulation of the element-wise version of Eq. \eqref{condition_Fullgain}, it can be shown that no such solution $\{\mathbf{T}^\star,\boldsymbol{\Phi}^\star\}$ can satisfy Eq. \eqref{condition_Fullgain} in general, leading to a contradiction. See Appendix A for details.}
\end{proof}
\subsection{Problem Formulation}
As shown in Proposition 1, achieving rainbow BF with full gain across all subcarriers for a given arbitrary frequency-direction mapping $\{\mathbf{u}_{\sf des},\mathbf{v}_{\sf des}\}$ is generally infeasible. Thus, our objective is to achieve the desired frequency-direction mapping approximately. To this end, we formulate a least-squares optimization problem to minimize the error norm between the JPTA beamformer and the desired rainbow beamformer as follows:
\begin{mini!}|l|[2]
{_{\boldsymbol{\alpha},\mathbf{T},\boldsymbol{\Phi}}}
{\hspace{-3 mm}\sum_{m=1}^M\left\lVert\mathbf{w}^{(m)}(\mathbf{T},\boldsymbol{\Phi}) - \alpha^{(m)}\mathbf{a}^{(m)}\left(u_{\sf des}^{(m)},v_{\sf des}^{(m)}\right)\right\rVert^2}
{\label{Rainbow BF Original Problem}}
{}
\addConstraint{\tau^{(n_x,n_y)}\in[0,\tau^{\sf max}], \;\forall n_x, n_y,}\label{TTD range}
\addConstraint{\phi^{(n_x,n_y)}\in[0,2\pi), \;\forall n_x, n_y,}\label{PS range}
\addConstraint{\big|\alpha^{(m)}\big|=1, \;\forall m,}\label{unit-modulus alpha}
\end{mini!}
 where \tcb{$\boldsymbol{\alpha}\triangleq[\alpha^{(1)},\cdots,\alpha^{(M)}]^{\sf T}$;} here, $\alpha^{(m)}$ represents the unit-modulus phase rotation coefficient for subcarrier $m$. Constraint \eqref{TTD range} limits the operating range of TTDs, while \eqref{PS range} accounts for the cyclic nature of phase shifts. Constraint \eqref{unit-modulus alpha} enforces the unit-modulus condition. The objective function is non-convex with respect to $\{\mathbf{T},\boldsymbol{\Phi},\boldsymbol{\alpha}\}$, making efficient optimization challenging. Moreover, for wideband systems with large antenna arrays, the number of optimization variables $(2N_{\sf rx}+M)$ grows significantly, leading to high computational complexity. To address this, we propose a joint alternating and decomposition-based optimization framework. \tcb{First, we derive closed-form optimal phase rotation coefficients $\boldsymbol{\alpha}^{\star}$ for given $\{\mathbf{T},\boldsymbol{\Phi}\}$. Second, we derive closed-form optimal PS values $\boldsymbol{\Phi}^{\star}$ for given $\{\mathbf{T},\boldsymbol{\alpha}\}$. Then, to obtain optimized TTD values $\mathbf{T}^\star$ for given $\boldsymbol{\alpha}$, we decompose the original high-dimensional problem into multiple sub-problems and efficiently solve them by parallel 1D line searches. By iteratively updating the optimization variables, our algorithm outputs an optimized rainbow beamformer.}

\subsection{Alternating and Decomposition-based Optimization}
\tcb{We re-express the objective function and define $F(\boldsymbol{\alpha},\mathbf{T},\boldsymbol{\Phi})$ as follows:
\begin{align}
&\sum_{m=1}^M\left\lVert\mathbf{w}^{(m)}(\mathbf{T},\boldsymbol{\Phi}) - \alpha^{(m)}\mathbf{a}^{(m)}\left(u_{\sf des}^{(m)},v_{\sf des}^{(m)}\right)\right\rVert^2
\nonumber\\
&\overset{(\rm{a})}{=}2MN_{\sf rx}
\nonumber\\
&-2\underbrace{\sum_{m=1}^{M}\Re\left(\left(\alpha^{(m)}\right)^*\mathbf{a}^{(m)}\left(u_{\sf des}^{(m)},v_{\sf des}^{(m)}\right)^{\sf H}\mathbf{w}^{(m)}(\mathbf{T},\boldsymbol{\Phi})\right)}_{\triangleq F(\boldsymbol{\alpha},\mathbf{T},\boldsymbol{\Phi})}.
\label{Definition_NewObejective}
\end{align} In step (a), we use the fact that $\big\lVert\mathbf{w}^{(m)}(\mathbf{T},\boldsymbol{\Phi})\big\rVert^2 = \big\lVert \alpha^{(m)}\mathbf{a}^{(m)}\big(u_{\sf des}^{(m)},v_{\sf des}^{(m)}\big)\big\rVert^2 = N_{\sf rx}$. Since this constant term does not affect the optimization solution, maximizing $F\left(\boldsymbol{\alpha},\mathbf{T},\boldsymbol{\Phi}\right)$ yields the same solutions as the original problem. For given $\{\mathbf{T}, \boldsymbol{\Phi}\}$, $\alpha^{(m)}$ can be any unit-modulus complex number for each subcarrier. Hence, maximizing $F\left(\boldsymbol{\alpha},\mathbf{T},\boldsymbol{\Phi}\right)$ can be decomposed into $M$ independent subproblems as the following subcarrier-wise optimization problem for $m\in\{1,\cdots,M\}$:
\begin{maxi!}|l|[2]
{_{\alpha^{(m)}}}
{\Re\left(\left(\alpha^{(m)}\right)^* \mathbf{a}^{(m)}\left(u_{\sf des}^{(m)},v_{\sf des}^{(m)}\right)^{\sf H}\mathbf{w}^{(m)}(\mathbf{T},\boldsymbol{\Phi})\right)}
{\label{alpha_opt}}
{}
\addConstraint{\big|\alpha^{(m)}\big|=1.}
\end{maxi!}
By defining $z^{(m)}\triangleq\mathbf{a}^{(m)}\left(u_{\sf des}^{(m)},v_{\sf des}^{(m)}\right)^{\sf H}\mathbf{w}^{(m)}(\mathbf{T},\boldsymbol{\Phi})$ for notational brevity, the objective function in problem \eqref{alpha_opt} can be rewritten as
\begin{align}
\Re\left(\left(\alpha^{(m)}\right)^*z^{(m)}\right)&=\Re\left(\left\vert\alpha^{(m)}\right\vert e^{-j\angle\alpha^{(m)}}\left\vert z^{(m)}\right\vert e^{j\angle z^{(m)}} \right)
\nonumber\\
&= \left\vert\alpha^{(m)} z^{(m)}\right\vert\cos(\angle z^{(m)}-\angle\alpha^{(m)}).
\end{align}
Since $\left\vert\alpha^{(m)}\right\vert = 1$ and $\left\vert z^{(m)}\right\vert$ is fixed for given $\{\mathbf{T}, \boldsymbol{\Phi}\}$, the objective function is maximized when $\cos(\angle z^{(m)}-\angle\alpha^{(m)}) = 1$, which occurs when $\angle\alpha^{(m)} = \angle z^{(m)}$. This gives us \tcr{a} closed-form optimal solution:
\begin{align}
\label{optimal alpha}
\alpha^{\star(m)} = \exp\left[j\angle\left(\mathbf{a}^{(m)}\left(u_{\sf des}^{(m)},v_{\sf des}^{(m)}\right)^{\sf H}\mathbf{w}^{(m)}(\mathbf{T},\boldsymbol{\Phi})\right)\right],\,\forall m.
\end{align}}

\tcb{For notational simplicity, we define $\delta_m^{(n_x,n_y)}\triangleq -\angle\alpha^{(m)}+\frac{f_m}{f_{\sf c}}\left((n_x-1)u_{\sf des}^{(m)}+(n_y-1)v_{\sf des}^{(m)}\right)$. Using this definition, $F\left(\boldsymbol{\alpha},\mathbf{T},\boldsymbol{\Phi}\right)$ can be rewritten as follows:
\begin{align}
&F\left(\boldsymbol{\alpha},\mathbf{T},\boldsymbol{\Phi}\right)
\nonumber\\
&=\sum_{n_x=1}^{N_x}\sum_{n_y=1}^{N_y}\Re\left(\sum_{m=1}^{M}e^{j\delta^{(n_x,n_y)}_m}e^{j\phi^{(n_x,n_y)}}e^{-j2\pi f_m\tau^{(n_x,n_y)}}\right)
\nonumber\\
&=\sum_{n_x=1}^{N_x}\sum_{n_y=1}^{N_y}\Re\Bigg(e^{j\phi^{(n_x,n_y)}}\underbrace{\sum_{m=1}^M e^{j\left\{\delta^{(n_x,n_y)}_m-2\pi f_m \tau^{(n_x,n_y)}\right\}}}_{\triangleq S\left(\boldsymbol{\alpha},\tau^{(n_x,n_y)}\right)}\Bigg)
\nonumber\\
&=\sum_{n_x=1}^{N_x}\sum_{n_y=1}^{N_y}\left\vert S\left(\boldsymbol{\alpha},\tau^{(n_x,n_y)}\right)\right\vert
\nonumber\\
&\qquad\qquad\qquad\times\cos\left(\phi^{(n_x,n_y)}+\angle S\left(\boldsymbol{\alpha},\tau^{(n_x,n_y)}\right)\right).
\label{Objective function manipulation}
\end{align}
The optimal phase shift that makes the cosine term equal in the final line of Eq. \eqref{Objective function manipulation} to one is given by
\begin{align}
\phi^{\star(n_x,n_y)}=-\angle S\left(\boldsymbol{\alpha},\tau^{(n_x,n_y)}\right), \forall n_x, n_y.
\label{optimal PS}
\end{align}
Hence, for a given $\boldsymbol{\alpha}$, the optimal phase shifter values $\boldsymbol{\Phi}$ are determined by $\mathbf{T}$. Consequently, our objective now becomes maximizing $\sum_{n_x=1}^{N_x}\sum_{n_y=1}^{N_y}\left\vert S\left(\boldsymbol{\alpha},\tau^{(n_x,n_y)}\right)\right\vert$. Note that $S\left(\boldsymbol{\alpha},\tau^{(n_x,n_y)}\right)$ depends only on the delay $\tau^{(n_x,n_y)}$ for a given $\boldsymbol{\alpha}$. Therefore, we can decompose the optimization problem into antenna element-wise subproblems as follows:}
\tcb{\begin{maxi!}|l|[2]
{_{\tau^{(n_x,n_y)}}}
{\left\vert S\left(\boldsymbol{\alpha},\tau^{(n_x,n_y)}\right)\right\vert}
{\label{[Rainbow BF Design] 1D search sub-problem]}}
{}
\addConstraint{\tau^{(n_x,n_y)}\in[0,\tau^{\sf max}]}.
\end{maxi!}
The solution of sub-problem \eqref{[Rainbow BF Design] 1D search sub-problem]} can be obtained through 1D line searches with respect to $\tau^{(n_x,n_y)}\in[0,\tau^{\sf max}]$ for all $n_x\in\{1,\cdots,N_x\}$ and $n_y\in\{1,\cdots,N_y\}$. Due to the independence of each sub-problem, we can efficiently solve \eqref{[Rainbow BF Design] 1D search sub-problem]} for all antenna elements using off-the-shelf parallel computing software such as \tcr{the} MATLAB Parallel Computing Toolbox \cite{matlabparallel}.} 

\tcb{We summarize the optimization process in Algorithm 1. Specifically, when $\boldsymbol{\alpha}$ is fixed, the optimal TTD and PS values $\{\mathbf{T}^\star,\boldsymbol{\Phi}^\star\}$ can be found by solving problem \eqref{[Rainbow BF Design] 1D search sub-problem]} with a parallel 1D line search and by applying Eq. \eqref{optimal PS}. Subsequently, the optimal $\boldsymbol{\alpha}^\star$ is updated through Eq. \eqref{optimal alpha}. This process is repeated until $F(\boldsymbol{\alpha},\mathbf{T},\boldsymbol{\Phi})$ converges, yielding the optimized rainbow beamformer $\mathbf{w}^{(m)}(\mathbf{T}^\star,\boldsymbol{\Phi}^\star)$.}

\tcb{\begin{remark}{\rm \textbf{(Convergence of Algorithm 1):}}
The convergence of Algorithm 1 to a local optimum is guaranteed by the following two properties: (i) the objective function $F(\boldsymbol{\alpha}, \mathbf{T}, \mathbf{\Phi})$ is bounded below by $MN_{\sf rx}$ as shown in Eq. \eqref{Definition_NewObejective} and (ii) each iteration of the alternating optimization monotonically increases $F(\boldsymbol{\alpha}, \mathbf{T}, \mathbf{\Phi})$, where the alternating optimization alternates between updating $\{\mathbf{T}^\star,\boldsymbol{\Phi}^\star\}$ for a given $\boldsymbol{\alpha}$ and updating $\boldsymbol{\alpha}^\star$ for given $\{\mathbf{T},\boldsymbol{\Phi}\}$ \cite{convergence}.
\end{remark}}

\tcb{\begin{remark}{\rm \textbf{(Computational Complexity of Algorithm 1):}}
The computational complexity of the proposed joint alternating and decomposition-based optimization algorithm (Algorithm 1) \tcr{can be} analyzed as follows. For each antenna element $(n_x, n_y)$, the 1D line search in problem \eqref{[Rainbow BF Design] 1D search sub-problem]} requires evaluating $|S(\boldsymbol{\alpha}, \tau^{(n_x,n_y)})|$ over $G = \lceil\tau_{\max}/\Delta\tau\rceil$ grid points, where $\Delta\tau$ denotes the grid resolution. Since $S(\boldsymbol{\alpha}, \tau^{(n_x,n_y)}) = \sum_{m=1}^{M} e^{j\{\delta_m^{(n_x,n_y)} - 2\pi f_m \tau^{(n_x,n_y)}\}}$, each grid point evaluation has complexity $\mathcal{O}(M)$. Thus, the single antenna element optimization requires $\mathcal{O}(G M)$ \tcr{operations}. Therefore, the overall complexity of Algorithm 1 is $\mathcal{O}(G M I N_{\sf rx})$, where $I$ denotes the number of iterations until convergence. 
\end{remark}}

\begin{algorithm}[!t]
\caption{\tcb{Joint Alternating and Decomposition-based Rainbow BF Optimization Algorithm}}
\begin{algorithmic}[1]
\State \textbf{Input}: \tcb{$\tau^{\sf max}$ and frequency-direction mapping $\{\mathbf{u}_{\sf des},\mathbf{v}_{\sf des}\}$}
\State \textbf{Initialize}: $\boldsymbol{\alpha}$
\While{$F(\boldsymbol{\alpha},\mathbf{T},\boldsymbol{\Phi})$ does not converge}
   \State \tcb{Solve problem \eqref{[Rainbow BF Design] 1D search sub-problem]} and apply Eq. \eqref{optimal PS} for all antenna elements to obtain optimal $\{\mathbf{T}^\star,\boldsymbol{\Phi}^\star\}$}
   \State Update $\boldsymbol{\alpha}^\star$ using Eq. \eqref{optimal alpha}
\EndWhile
\State \textbf{Output}: \tcb{Optimized rainbow beamformer $\mathbf{w}^{(m)}(\mathbf{T}^\star,\boldsymbol{\Phi}^\star)$}
\end{algorithmic}
\end{algorithm}

\subsection{Design for Frequency-Direction Mapping $\{\mathbf{u}_{\sf des},\mathbf{v}_{\sf des}\}$}
\tcb{We aim to design a frequency-direction mapping that maximizes the number of simultaneously serviceable users, thereby boosting the throughput. Since our rainbow BF optimization algorithm can accommodate arbitrary frequency-direction mapping as input parameters, this flexibility of the proposed 3D BF framework enables the deployment of diverse frequency-direction mapping approaches.}
    
\tcb{While countless frequency-direction mappings are possible, we present two categories of frequency-direction mapping design. The first approach is \emph{fixed frequency-direction mapping}. In this approach, as shown in Fig. \ref{Rainbow BF patterns}, beams with different desired beam directions according to frequency collectively cover the entire coverage area simultaneously, regardless of user locations, traffic demands, etc. In this approach, there is no need to update the beamformer as user locations or environmental conditions change, enabling a low-cost implementation with fixed TTD elements that have pre-optimized delay values and cannot be dynamically tunable post-fabrication \cite{yan2022energy}. The second approach is \emph{adaptive frequency-direction mapping}. In this approach, there is room to further improve throughput by utilizing user location information, regional user density information, and traffic demand information, etc., compared to the fixed mapping. However, since TTD need to be dynamically tunable, this leads to increased hardware cost and power consumption \cite{7394105,liu2018rotman,yan2022energy}.}
    
\tcb{Both design approaches must consider \tcr{the} bandwidth, satellite's orbital altitude, coverage, and beam width, etc. In this work, we focus on \tcr{a fixed frequency-direction mapping} for simplicity and leave adaptive mapping as future work. In Sec. V, we numerically assess the impact of different frequency-direction mappings on 3D rainbow BF performance, focusing on the fixed mapping.}

\section{Joint Subcarrier and Power Allocation}
While conventional OFDMA subcarrier and power allocation methods heavily rely on users' instantaneous channel SNR \tcb{information} \cite{kim2005joint}, the short channel coherence time in LEO SATCOM systems leads to rapid channel SNR fluctuations \cite{you2022beam}, thus making these terrestrial network methods ineffective. To address this issue, we present a JSPA algorithm for the proposed rainbow BF-based LEO SATCOM systems that leverage statistical and geometric CSI instead of instantaneous CSI.
\subsection{Achievable Rate Approximation}
\tcb{To eliminate the dependency on instantaneous channel SNR in the achievable rate expression of Eq. \eqref{Throughput model}, we derive an upper bound of ergodic rates \cite{you2022beam} and use it as an approximation of the achievable rate:}
\begin{align}
&\mathbb{E}\Bigg[\mathrm{log}_2\Bigg(1+\frac{p_k^{(m)}\big\lvert\mathbf{w}^{(m)}(\mathbf{T},\boldsymbol{\Phi})^{\sf H}\mathbf{h}_k^{(m)}\big\rvert^2}{N_{\sf rx}\sigma^2}\Bigg)\Bigg]
\nonumber\\
&=\mathbb{E}\Bigg[\mathrm{log}_2\Bigg(1+\frac{p_k^{(m)}\big\lvert g_k^{(m)}\big\rvert^2\big\lvert\mathbf{w}^{(m)}(\mathbf{T},\boldsymbol{\Phi})^{\sf H}\mathbf{a}^{(m)}(u_k,v_k)\big\rvert^2}{N_{\sf rx}\sigma^2}\Bigg)\Bigg]
\nonumber\\
&\overset{(\rm{a})}{\leq}\mathrm{log}_2\Bigg(1+\frac{p_k^{(m)}\mathbb{E}\big[\big\lvert g_k^{(m)}\big\rvert^2\big]\big|\mathbf{w}^{(m)}(\mathbf{T},\boldsymbol{\Phi})^{\sf H}\mathbf{a}^{(m)}(u_k,v_k)\big|^2}{N_{\sf rx}\sigma^2}\Bigg)
\nonumber\\
&=\mathrm{log}_2\Big(1+p_k^{(m)}\gamma_k^{(m)}\Big),
\end{align}
where step ${\rm (a)}$ follows from the concavity of the logarithmic function; average channel SNR $\gamma_k^{(m)}$ is defined as
\begin{align}
\gamma_k^{(m)}\triangleq\frac{\eta_k^{(m)}\big\lvert\mathbf{w}^{(m)}(\mathbf{T},\boldsymbol{\Phi})^{\sf H}\mathbf{a}^{(m)}(u_k,v_k)\big\rvert^2}{N_{\sf rx}\sigma^2}.
\end{align}
As the Rician factor goes to infinity ($\kappa_k \rightarrow \infty$), the channel becomes deterministic, making the approximation increasingly tight. Notably, $\gamma_k^{(m)}$ is expressed in terms of the \tcb{geometric} information (i.e., $(u_k,v_k)$) and statistical information (i.e., $\eta_k^{(m)}$) of the channels. In contrast to instantaneous channel SNR, statistical and \tcb{geometric} information is often readily available or easily estimated thanks to the quasi-deterministic nature of satellite movement along a predetermined orbit around the Earth \cite{roper2022beamspace}.
\subsection{Problem Formulation}
For a given time slot, we formulate the joint subcarrier and power allocation problem to maximize the system throughput as follows:
\begin{maxi!}|l|[2]
{_{\mathbf{B},\mathbf{P}}}
{\sum_{m=1}^M\sum_{k=1}^K b_k^{(m)}\Delta f\log_2\left(1+p_k^{(m)}\gamma_k^{(m)}\right)}
{\label{Original Subcarrier/power allocation}}
{}
\addConstraint{b_k^{(m)}\in\{0,1\},\;\forall k, m,}
\addConstraint{\sum_{k=1}^K b_k^{(m)} = 1,\;\forall m,}
\addConstraint{\sum_{m=1}^M p_k^{(m)} \leq P_k, \;\forall k,}
\addConstraint{p_k^{(m)}\geq 0, \;\forall k,m,}
\end{maxi!}
where $\mathbf{B}\in\mathbb{R}^{K\times M}$ contains subcarrier allocation indices $b_k^{(m)}\in\{0,1\}$; $\mathbf{P}\in\mathbb{R}^{K\times M}$ contains power allocation factor $p_k^{(m)}$ for user $k$ on subcarrier $m$; $P_k$ denotes the instantaneous power budget of user $k$. Note that this problem is non-convex due to the binary subcarrier allocation constraints, making it challenging to solve efficiently. For mathematical tractability, we relax $b_k^{(m)}$ into a continuous variable constrained within the interval [0,1].
With this relaxation, the original problem can be reformulated as:
\begin{maxi!}|l|[2]
{_{\mathbf{B},\mathbf{P}}}
{\sum_{m=1}^M\sum_{k=1}^K b_k^{(m)}\Delta f\log_2\left(1+p_k^{(m)}\gamma_k^{(m)}\right)}
{\label{Relaxed resource allocation problem}}
{}
\addConstraint{\sum_{k=1}^K b_k^{(m)} \leq 1, \;\forall m,}\label{sum of b leq 1}
\addConstraint{\sum_{m=1}^M p_k^{(m)} \leq P_k, \;\forall k,}\label{power budget constraint}
\addConstraint{b_k^{(m)}\geq 0,\,\, p_k^{(m)}\geq 0, \;\forall k,m.}\label{non-negative b, p}
\end{maxi!}
\tcr{Even with this continuous variable relaxation, we will show in the following subsection that this relaxation does not violate our exclusive subcarrier allocation assumption in Proposition 2.}
\subsection{Optimality Conditions for Problem \eqref{Relaxed resource allocation problem}}
We derive the KKT conditions, which are necessary conditions for the optimal solution of relaxed problem \eqref{Relaxed resource allocation problem}. First, we define the Lagrangian function as follows:
\begin{align}
&L(\mathbf{B},\mathbf{P})\triangleq-\sum_{m=1}^M\sum_{k=1}^K b_k^{(m)}\Delta f\log_2\left(1+p_k^{(m)}\gamma_k^{(m)}\right)\nonumber\\
&+\sum_{m=1}^M\lambda^{(m)}\Bigg(\sum_{k=1}^K b_k^{(m)}-1\Bigg)+\sum_{m=1}^M \mu_k\Bigg(\sum_{k=1}^K p_k^{(m)}-P_k\Bigg)
\nonumber\\
&-\sum_{m=1}^M\sum_{k=1}^K \rho_k^{(m)} b_k^{(m)}- \sum_{m=1}^M\sum_{k=1}^K \nu_k^{(m)} p_k^{(m)},
\end{align}
where $\lambda^{(m)}$, $\mu_k$, $\rho_k^{(m)}$, and $\nu_k^{(m)}$ are \tcr{Lagrange} multipliers. Taking the derivatives with respect to $b_k^{(m)}$ and $p_k^{(m)}$:
\begin{align}
&\frac{\partial}{\partial b_k^{(m)}}L(\mathbf{B},\mathbf{P})=-\Delta f \log_2\left(1+p_k^{(m)}\gamma_k^{(m)}\right)+\lambda^{(m)}-\rho_k^{(m)},
\\
&\frac{\partial}{\partial p_k^{(m)}}L(\mathbf{B},\mathbf{P})=-\frac{\Delta f \gamma_k^{(m)} b_k^{(m)}}{\ln2\left(1+p_k^{(m)}\gamma_k^{(m)}\right)}+\mu_k-\nu_k^{(m)}.
\end{align}
We can obtain the stationary conditions by setting $\frac{\partial}{\partial b_k^{(m)}}L(\mathbf{B},\mathbf{P})=0$ and $\frac{\partial}{\partial p_k^{(m)}}L(\mathbf{B},\mathbf{P})=0$ for all $k$ and $m$. The KKT conditions are listed as follows:
\begin{itemize}
    \item Stationary:
    \begin{align}
        &-\Delta f \log_2\left(1+p_k^{(m)}\gamma_k^{(m)}\right)+\lambda^{(m)} = \rho_k^{(m)}, \;\forall k,m,\label{stationary1} \\
        &-\frac{\Delta f \gamma_k^{(m)} b_k^{(m)}}{\ln2\left(1+p_k^{(m)}\gamma_k^{(m)}\right)}+\mu_k=\nu_k^{(m)}, \;\forall k,m,\label{stationary2}
    \end{align}

    \item Complementary slackness:
    \begin{align}
    &\lambda^{(m)}\Bigg(\sum_{k=1}^K b_k^{(m)}-1\Bigg)=0,\;\forall m,\label{complementary1}
    \\
    &\mu_k\Bigg(\sum_{m=1}^M p_k^{(m)}-P_k\Bigg)=0,\;\forall k,\label{complementary2}
    \\
    &\rho_k^{(m)} b_k^{(m)}=0, \;\forall k,m,\label{complementary3}
    \\
    &\nu_k^{(m)} p_k^{(m)} =0,\;\forall k,m,\label{complementary4}
    \end{align}
    
    \item Primal feasibility: \eqref{sum of b leq 1}, \eqref{power budget constraint}, \eqref{non-negative b, p},
    
    \item Dual feasibility:
    \begin{align}
    &\lambda^{(m)}\geq0,\;\forall m,\label{dual1}
    \\
    &\mu_k\geq0,\;\forall k,\label{dual2}
    \\
    &\rho_k^{(m)}\geq0,\;\forall k,m,\label{dual3}
    \\
    &\nu_k^{(m)}\geq0,\;\forall k,m.\label{dual4}
    \end{align}
\end{itemize}
We proceed with the mathematical manipulation of these KKT conditions by combining the intertwined conditions.
\begin{itemize}
    \item From Eqs. \eqref{stationary1} and \eqref{complementary3}:    
    \begin{align}
        \hspace{-4mm}b_k^{(m)}\left\{\Delta f \log_2\left(1+p_k^{(m)}\gamma_k^{(m)}\right)-\lambda^{(m)}\right\}=0,\;\forall k,m, \label{combined1}
    \end{align}
    
    \item From Eqs. \eqref{stationary1} and \eqref{dual3}:
        \begin{align}
            \Delta f \log_2\left(1+p_k^{(m)}\gamma_k^{(m)}\right)-\lambda^{(m)}\leq0,\;\forall k,m, \label{combined2}
        \end{align}

    \item From Eqs. \eqref{stationary2} and \eqref{complementary4}:
    \begin{align}
         p_k^{(m)}\left\{\frac{\Delta f \gamma_k^{(m)} b_k^{(m)}}{\ln2\left(1+p_k^{(m)}\gamma_k^{(m)}\right)}-\mu_k\right\}=0, \;\forall k,m, \label{combined3}
    \end{align}
    
    \item From Eqs. \eqref{stationary2} and \eqref{dual4}:
    \begin{align}
         \frac{\Delta f \gamma_k^{(m)} b_k^{(m)}}{\ln2\left(1+p_k^{(m)}\gamma_k^{(m)}\right)}-\mu_k\leq0, \;\forall k,m, \label{combined4}
    \end{align}

\end{itemize}

\tcb{From Eqs. \eqref{combined1} and \eqref{combined2}, we derive optimality conditions for subcarrier allocation by analyzing two distinct cases. When subcarrier $m$ is not allocated to user $k$ (i.e., $b_k^{(m)}=0$), Eq. \eqref{combined1} is automatically satisfied. Then, $\Delta f \log_2\big(1+p_k^{(m)}\gamma_k^{(m)}\big)\leq\lambda^{(m)}$ holds from Eq. \eqref{combined2}. Conversely, when subcarrier $m$ is allocated to user $k$ (i.e., $b_k^{(m)}>0$), Eq. \eqref{combined1} requires $\Delta f \log_2\big(1+p_k^{(m)}\gamma_k^{(m)}\big)=\lambda^{(m)}$. These two conditions hold for all $k$ when $\lambda^{(m)}=\underset{k}{\mathrm{max}}\,\Delta f\log_2\left(1+p_k^{(m)}\gamma_k^{(m)}\right)$; hence, the \tcr{Lagrange} multiplier $\lambda^{(m)}$ represents the maximum rate for subcarrier $m$. Therefore, the $m$-th subcarrier is allocated to user $k^{\star(m)}$ to achieve this maximum rate $\lambda^{(m)}$, such that
\begin{align}
k^{\star(m)}=\underset{k}{\mathrm{argmax}}\,\Delta f\log_2\left(1+p_k^{(m)}\gamma_k^{(m)}\right)=\underset{k}{\mathrm{argmax}}\,p_k^{(m)}\gamma_k^{(m)}.
\label{Optimality condition 1}
\end{align}} As mentioned earlier, $b_k^{(m)}$ is relaxed to a continuous real variable for mathematical tractability. Therefore, there is no guarantee that $k^{\star(m)}$ satisfying \eqref{Optimality condition 1} is unique for a given subcarrier $m$. \tcb{However, the following proposition shows that the optimal value of our relaxed problem is preserved even if we allocate each subcarrier to only one user satisfying \eqref{Optimality condition 1}.}
\begin{prop}\tcb{Let $\mathcal{K}_m$ denote the set of users satisfying \eqref{Optimality condition 1} for subcarrier $m$. Although we derive the optimal conditions by relaxing $b_k^{(m)}$ to continuous values in problem \eqref{Relaxed resource allocation problem}, enforcing the original binary constraint $b_k^{(m)}\in\{0,1\}$ with exclusive subcarrier allocation does not compromise the optimal value \cite{kim2005joint}.}
\end{prop}
\begin{proof}
For subcarrier $m$, $b_k^{(m)}>0$ if $k\in\mathcal{K}_m$, and $b_k^{(m)}=0$ if $k\notin\mathcal{K}_m$. Thus, $\sum_{k=1}^K b_k^{(m)}=\sum_{k\in\mathcal{K}_m} b_k^{(m)}$. When subcarrier $m$ is allocated to a user $k$, then  $\Delta f\log_2\big(1+p_k^{(m)}\gamma_k^{(m)}\big)=\lambda^{(m)}>0$ for all $k\in\mathcal{K}_m$. Since $\sum_{k\in\mathcal{K}_m}b_k^{(m)}=1$ from Eq. \eqref{complementary1}, the following equality holds:
\begin{align}
\sum_{k\in\mathcal{K}_m} b_k^{(m)}\underbrace{\Delta f\log_2(1+p_k^{(m)}\gamma_k^{(m)})}_{=\lambda^{(m)}}=\lambda^{(m)}\sum_{k\in\mathcal{K}_m} b_k^{(m)}=\lambda^{(m)}.
\end{align}
\tcb{Therefore, the maximum rate of subcarrier $m$ remains unchanged as $\lambda^{(m)}$ regardless of any combination of $b_k^{(m)}$ values satisfying $\sum_{k\in\mathcal{K}_m}b_k^{(m)}=1$, allowing exclusive subcarrier allocation as in the original discrete optimization problem.}
\end{proof} 

From Eqs. \eqref{combined3} and \eqref{combined4}, if $k$-th user's power is not allocated to $m$-th subcarrier, i.e., $p_k^{(m)}=0$, then $\frac{\Delta f \gamma_k^{(m)} b_k^{(m)}}{\ln2\left(1+p_k^{(m)}\gamma_k^{(m)}\right)}\leq\mu_k$, and if $k$-th user's power is allocated to $m$-th subcarrier, i.e., $p_k^{(m)}>0$, then $\frac{\Delta f \gamma_k^{(m)} b_k^{(m)}}{\ln2\left(1+p_k^{(m)}\gamma_k^{(m)}\right)}=\mu_k$. This condition is satisfied if $k$-th user's power is allocated to subcarrier $m_k^\star$ such that
\begin{align}
m_k^\star=\underset{m}{\mathrm{argmax}}\,\frac{\Delta f \gamma_k^{(m)} b_k^{(m)}}{\ln2\left(1+p_k^{(m)}\gamma_k^{(m)}\right)}=\underset{m}{\mathrm{argmax}}\,\frac{\gamma_k^{(m)} b_k^{(m)}}{1+p_k^{(m)}\gamma_k^{(m)}}.
\label{Optimality condition 2}
\end{align}
While Eq. \eqref{Optimality condition 2} establishes the conditions for power allocation to subcarriers, it does not explicitly describe the optimal power allocation amounts. Note that the original problem \eqref{Original Subcarrier/power allocation} is a convex optimization problem for a given subcarrier allocation $\mathbf{B}$, and the optimal $\mathbf{P}$ can be obtained through \tcr{a} water-filling \cite{kim2005joint}. Specifically, from Eqs. \eqref{stationary2} and \eqref{complementary4}, the optimal power allocation is given by
\begin{align}
p_k^{\star(m)}=\left[\frac{\Delta f b_k^{(m)}}{\mu_k\ln2}-\frac{1}{\gamma_k^{(m)}}\right]^+.
\end{align}
Defining the water-filling level as
$\Delta_k\triangleq\frac{\Delta f b_k^{(m)}}{\mu_k\ln2}$, we ensure it satisfies the power budget constraint in Eq. \eqref{power budget constraint}. Under the condition that $p_k^{(m)}>0$ for subcarrier allocation, the optimality condition in Eq. \eqref{Optimality condition 2} can be reformulated as:
\begin{align}
m_k^\star=\underset{m}{\mathrm{argmax}}\,\frac{\gamma_k^{(m)} b_k^{(m)}}{1+\Big(\Delta_k-\frac{1}{\gamma_k^{(m)}}\Big)\gamma_k^{(m)}}=\underset{m}{\mathrm{argmax}}\,b_k^{(m)}.
\label{optimality condition 3}
\end{align}
Since Proposition 2 establishes that the exclusive subcarrier allocation preserves objective function value, Eq. \eqref{optimality condition 3} trivially implies that power should be allocated only to subcarriers with nonzero $b^{(m)}_k$. Thus, the necessary condition that requires attention is Eq. \eqref{Optimality condition 1}. However, finding a solution remains challenging due to the interdependence of subcarrier and power allocations in Eq. \eqref{Optimality condition 1}. Therefore, an explicit algorithm is needed to jointly optimize both subcarrier and power allocation.
\subsection{Greedy Algorithm for Joint Subcarrier \& Power Allocation}
\begin{algorithm}[!t]
\caption{Greedy Algorithm for JSPA}\label{Algorithm_greedy}
\begin{algorithmic}[1]
\State \textbf{Input}: $\{\gamma_k^{(m)}\}_{\forall k, m}$, $\{P_k\}_{\forall k}$
\State \textbf{Initialize}: $\mathbf{B}\leftarrow \mathbf{0}_{K\times M}$, $\mathbf{P}\leftarrow \mathbf{0}_{K\times M}$
\State Define unallocated subcarrier set $\mathcal{M}_{\sf un}\leftarrow \{1,\cdots,M\}$
\While{$\mathcal{M}_{\sf un}\neq\emptyset$}
   \State Randomly select a subcarrier $m$ from $\mathcal{M}_{\sf un}$
   \State Temporarily allocate $m$ to all users and update $\mathbf{P}$
   \State \tcb{Select the single user $k^{\star(m)}=\underset{k}{\mathrm{argmax}}\,p_k^{(m)}\gamma_k^{(m)}$}
   \State \tcb{Exclusively allocate subcarrier $m$ to user $k^{\star(m)}$}
   \State $\mathcal{M}_{\sf un}\leftarrow\mathcal{M}_{\sf un}\setminus\{m\}$
\EndWhile
\State Update $\mathbf{P}$ using water-filling
\State \textbf{Output}: Optimized $\{\mathbf{B}^\star, \mathbf{P}^{\star}\}$ 
\end{algorithmic}
\end{algorithm}
We propose a greedy algorithm to efficiently address the intertwined condition in Eq. \eqref{Optimality condition 1}, as outlined in Algorithm 2. \tcb{The key insight of our greedy approach is to employ an iterative evaluation strategy. At each stage of the greedy selection process, a subcarrier $m$ is randomly chosen from the unallocated subcarrier set and temporarily assigned to all users. Then, power allocation is performed through \tcr{a} water-filling method, and the user $k^{\star(m)}$ that satisfies Eq. \eqref{Optimality condition 1} is identified. Subsequently, subcarrier $m$ is allocated to user $k^\star$. \tcr{This} selection process repeats until all subcarriers are allocated. This strategy considers channel quality and individual power constraints jointly, rather than simply selecting the user with the best channel condition (\tcr{which is} typically downlink optimal \cite{kim2005joint}).}

\tcb{\begin{remark}{\rm \textbf{(Computational Complexity of Algorithm 2):}} The computational complexity of the proposed JSPA algorithm consists of two components: (i) best user search requiring $\mathcal{O}(MK)$ operations across $M$ iterations, and (ii) water-filling computation requiring $\mathcal{O}(\sum_{k=1}^{K} S_{i,k} \log S_{i,k})$ at the $i$-th iteration \cite{khakurel2014generalized}, where $S_{i,k}$ denotes the number of subcarriers allocated to user $k$ with $\sum_{k=1}^K S_{i,k} = i$. The overall complexity is $\mathcal{O}(MK + \sum_{i=1}^{M} \sum_{k=1}^{K} S_{i,k} \log S_{i,k})$. In the worst case, when subcarriers are uniformly distributed among users (i.e., $S_{i,k} \approx i/K$ for all $k$), the complexity becomes approximately $\mathcal{O}(MK + M^2 \log \frac{M}{K})$ from the integral approximation of a sum, which is significantly more efficient than exhaustive search requiring $\mathcal{O}(K^M)$.
\end{remark}}

\section{Numerical Results}

\begin{table}[!t]
\label{Simulation Parameter Table}
\centering
\caption{Simulation Parameters}
\renewcommand{\arraystretch}{1.5}
\begin{tabular}{P{0.13\linewidth}P{0.16\linewidth}P{0.3\linewidth}P{0.25\linewidth}}
\Xhline{1.5pt} 
\textbf{} & \textbf{Abbreviation} & \textbf{Parameter} & \textbf{Value}\\
\Xhline{1.5pt} 
\cellcolor{White} & $f_{\sf c}$ & Center frequency & 14 GHz\\
\rowcolor{Gray}
\cellcolor{White} & - & Bandwidth & \tcb{0.7 - 2.1} GHz\\
\cellcolor{White} & $M$ & Number of subcarriers & 1,024\\
\rowcolor{Gray}
\multirow{-4}[0]{*}{\cellcolor{White}\textbf{Channel}} & $\kappa_k$ & Rician factor & 10 dB\\
\hline
\cellcolor{White} & $N_{\sf sat}$ & Number of antennas & 64 (8×8 URA)\\
\rowcolor{Gray}
\cellcolor{White} & $G_{\sf sat}$ & Antenna gain & 0 dBi\\
\multirow{-3}[0]{*}{\cellcolor{White}\textbf{Satellite}} & - & Number of RF chain & 1\\
\hline
\rowcolor{Gray}
\cellcolor{White} & $K$ & Number of users & 2 - 256\\
\cellcolor{White} & - & Number of antennas & 1\\
\rowcolor{Gray}
\cellcolor{White} & $G_{\sf ut}$ & Antenna element gain & 43.2 dBi\\
\multirow{-4}[0]{*}{\cellcolor{White}\shortstack{\textbf{User}\\\textbf{terminal}}} & $P_k$ & Power budget & 23 dBm\\
\Xhline{1.5pt} 
\end{tabular}
\renewcommand{\arraystretch}{1}
\end{table}

\tcb{To evaluate the performance of the proposed system, we conduct realistic simulations considering 3D geometry and Earth's curvature, where the satellite orbits at \num{500} km altitude with ground users randomly distributed within a \num{500} km radius area on the Earth's surface. The average channel power is given by \cite{you2020massive, you2022beam}
\begin{equation}
\eta_k^{(m)}=G_{\sf sat}G_{\sf ut}\left(\frac{c}{4\pi f_m d_k}\right)^2,\,\forall k,m,
\end{equation}
where $G_{\sf sat}$ and $G_{\sf ut}$ denote the satellite and user terminal antenna gains, respectively.} $d_k$ represents the distance between the satellite and $k$-th user; $c$ is the speed of light. \tcb{The noise power is $\sigma^2=\mathcal{B}\Delta fT$, where $\mathcal{B}$ is the Boltzmann constant; $T$ is the noise temperature set to $T=290$ K. The remaining simulation parameters are listed in Table I.}
\subsection{Performance of Proposed Rainbow BF Optimization Algorithm and Beam Footprints}
\begin{figure*}[!ht]
\centering
\includegraphics[width=1\linewidth]{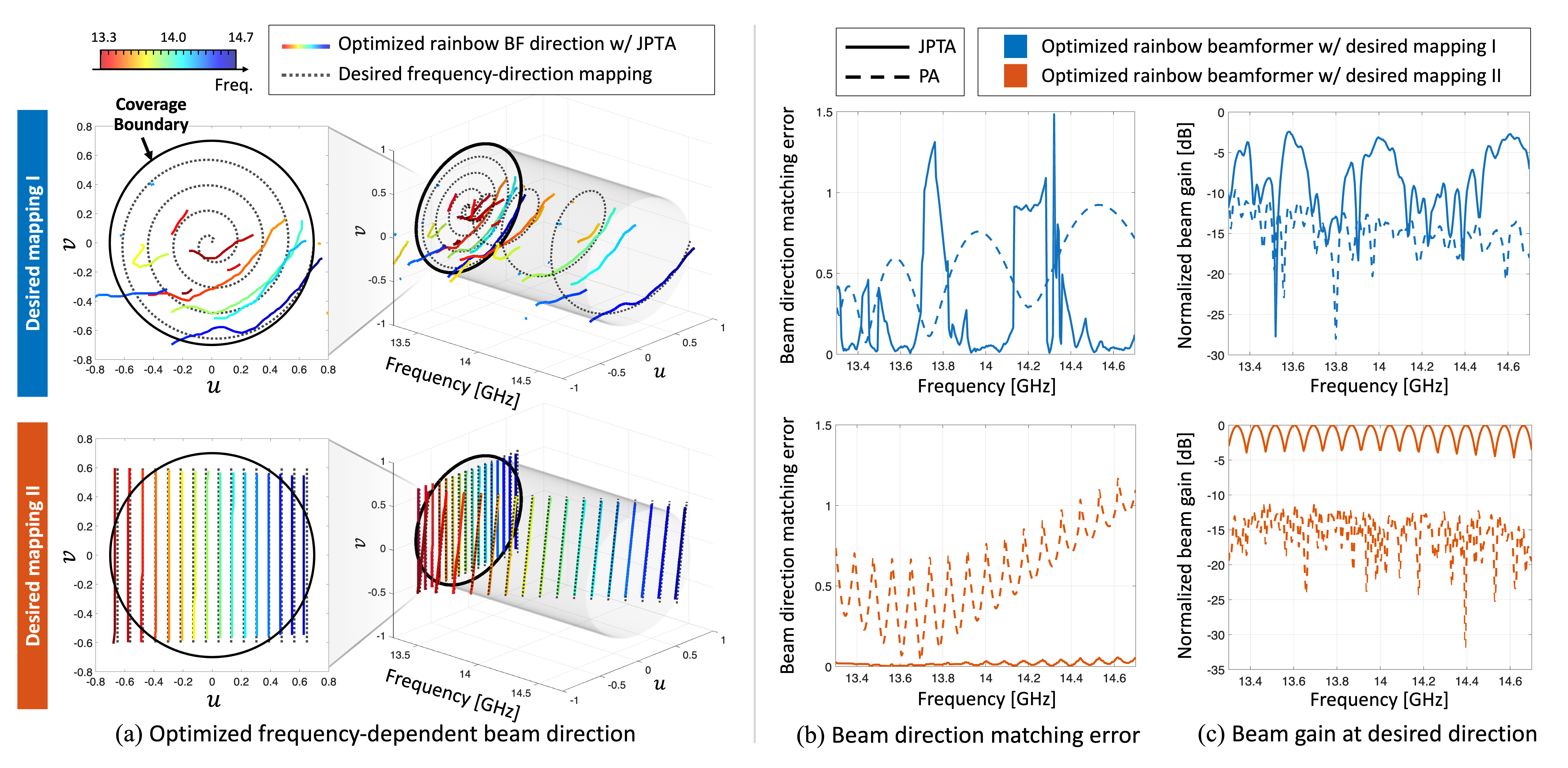}
\caption{\tcb{Beamforming performance comparison under two different frequency-direction mappings. (a): frequency-dependent beam direction in frequency-UV plots and UV-plane projections, (b): beam direction matching error with JPTA and PA, (c): beam gain achieved at desired directions across frequencies.}}
\label{[Simulation]rainbowBFPerformance}
\end{figure*}
\begin{figure}[!ht]
\centering
\includegraphics[width=1\linewidth]{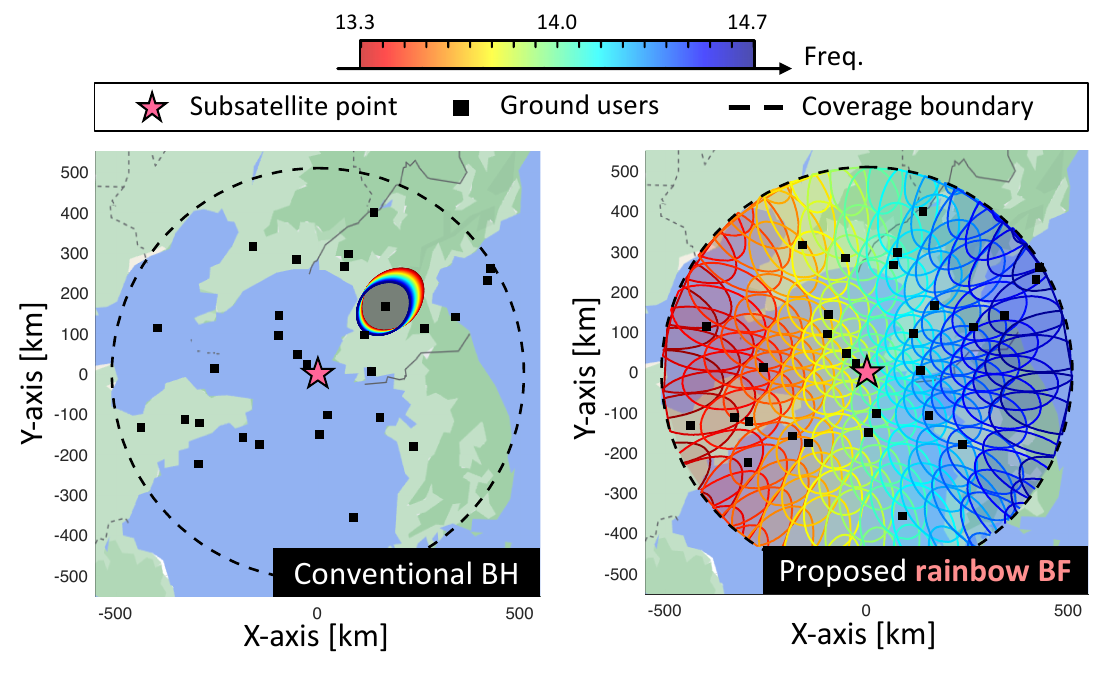}
\caption{\tcb{3 dB beam footprints on Earth's surface (top view) at a specific time slot comparing BH (without beam-squint effects) showing narrow beam footprints and rainbow BF fully covering the entire coverage area with $K=32$ users. Beam colors represent subcarrier frequencies. A subset of subcarriers is sampled for visualization.}}
\label{[Simulation]Footprints}
\end{figure}

We validate the feasibility and evaluate the numerical performance of the proposed rainbow BF optimization algorithm. \tcbb{In the 1D line search for rainbow beamformer optimization, we performed a grid search over \num{2000} points ranging from \num{0} to \num{50} ns with a step size of \num{25} ps.} Among numerous potential frequency-direction mappings, we analyze two representative \tcb{fixed mappings} shown in Fig. \ref{Rainbow BF patterns}. The frequency-dependent beam directions of the optimized JPTA beamformer and their corresponding UV-plane projections are illustrated in Fig.~\ref{[Simulation]rainbowBFPerformance} (a). \tcb{Here, the beam direction for the $m$-th subcarrier of the optimized beamformer $(u_{\sf opt}^{(m)},v_{\sf opt}^{(m)})$ is measured as the maximum beam gain direction, given by
\begin{align}
\left(u_{\sf opt}^{(m)},v_{\sf opt}^{(m)}\right)=\underset{u, v}{\mathrm{argmax}}\,\left\vert \mathbf{w}(\mathbf{T},\boldsymbol{\Phi})^{\sf H}\mathbf{a}^{(m)}(u,v)\right\vert^2, \forall m.
\end{align}} The desired mappings I and II are designed to scan the coverage area using spiral and multiple parallel line trajectories, respectively. \tcb{The numerical results show that desired mapping II achieves \tcr{better} frequency-dependent beam direction matching than mapping I.} Fig.~\ref{[Simulation]rainbowBFPerformance} (b) and (c) compare the beam direction matching error and beam gain degradation between conventional PA and JPTA. \tcb{The beam direction matching errors are measured as $\Vert [u_{\sf opt}^{(m)}-u_{\sf des}^{(m)},v_{\sf opt}^{(m)}-v_{\sf des}^{(m)}]\Vert$ for each subcarrier $m$. The PA case is obtained by setting $\tau_{\sf max}=0$ in Algorithm 1.} \tcb{While conventional PA suffers from significant beam gain loss and beam direction matching error for both frequency-direction mappings due to the limitation of frequency-flat phase shifts of the PS, JPTA performance varies with different mappings in terms of beam direction matching and BF gain degradation. Notably, JPTA with mapping II achieves superior performance by preserving high beam gains toward the desired directions. In the following, we focus on mapping II for the performance evaluation of rainbow BF.}

The 3 dB beam footprints on the Earth's surface for both BH and rainbow BF are illustrated in Fig. \ref{[Simulation]Footprints}. Conventional BH \tcb{with PA} exhibits nearly aligned beams across all subcarriers with beam-squint, which limits service coverage to a small area. In contrast, the proposed rainbow BF with mapping II steers beams at different subcarrier frequencies toward spatially diverse directions, thereby enabling simultaneous signal reception from multiple users across the entire coverage area.


\subsection{Throughput Performance Evaluation}
\begin{figure*}[!ht]
\centering
\includegraphics[width=1\linewidth]{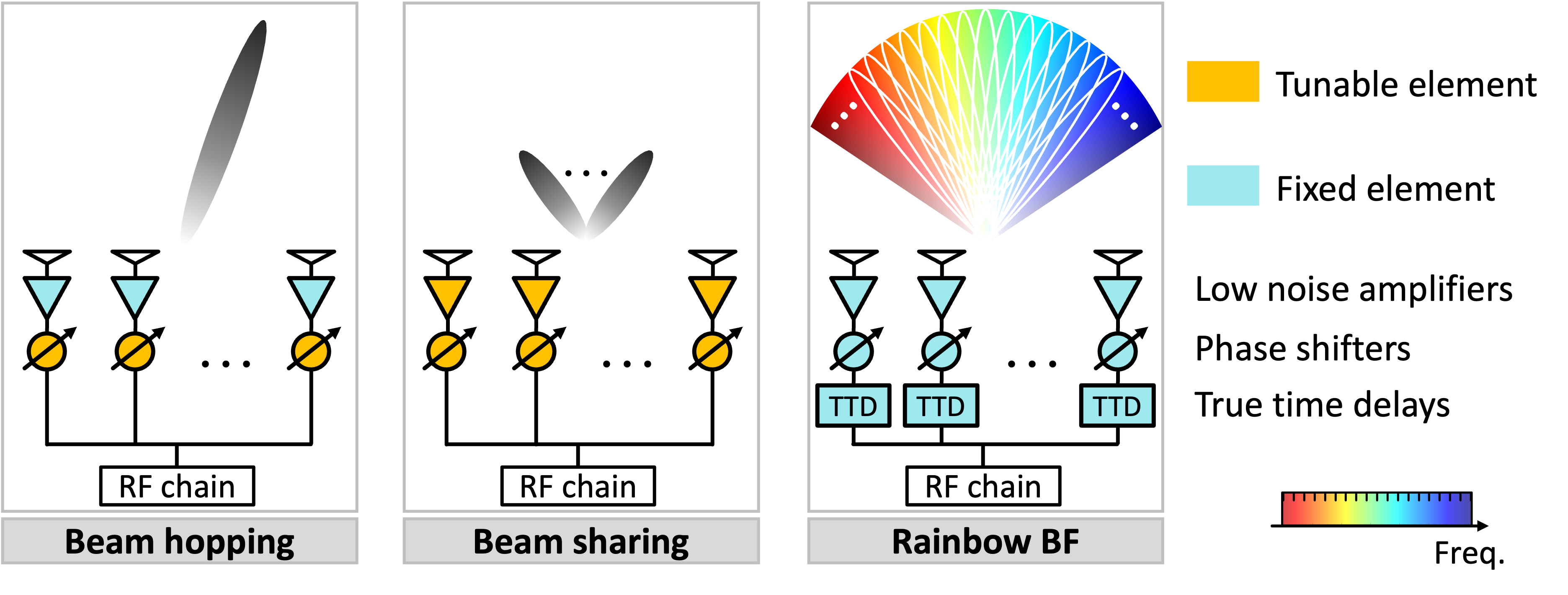}
\caption{\tcb{Antenna array structures for numerical evaluation of BH, beam sharing, and proposed rainbow BF. For BH, beam-squint can be eliminated by using TTD instead of PS.}}
\label{[Simulation]ArrayStructures}
\vspace{0mm}
\end{figure*}

\begin{figure*}[!ht]
\centering
\includegraphics[width=1\linewidth]{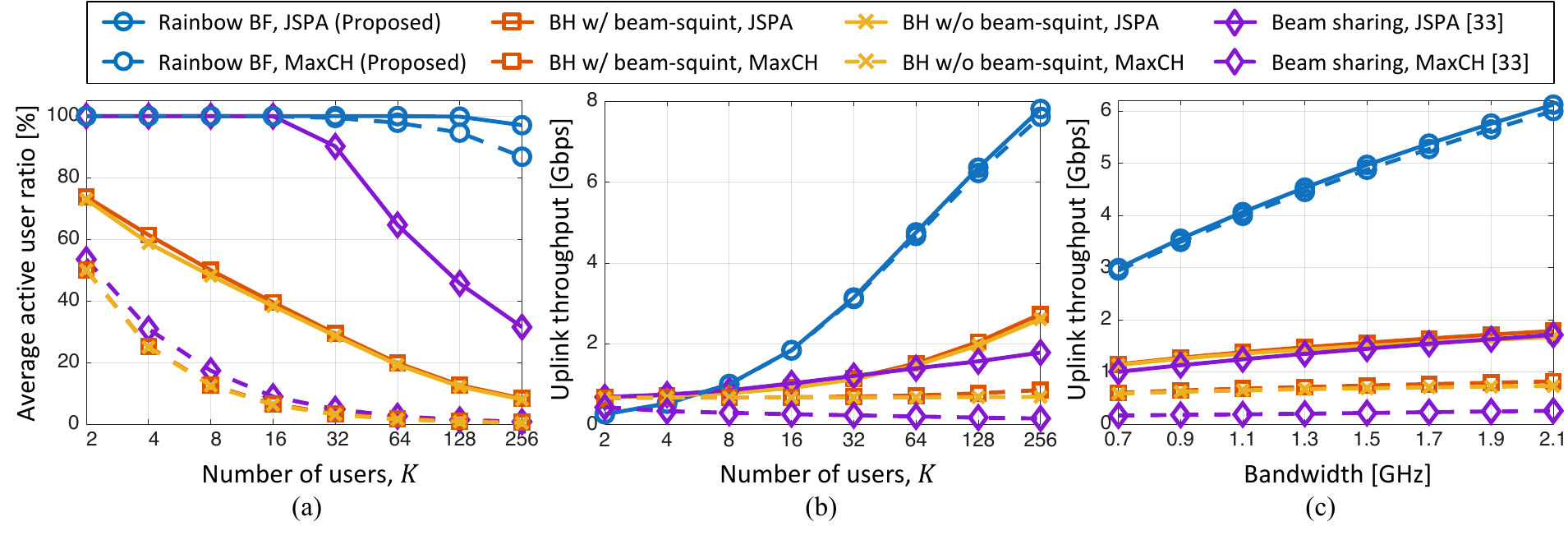}
\caption{\tcb{Performance comparison of conventional BH (with and without beam-squint), beam sharing, and proposed rainbow BF: (a) average active user ratio versus number of users, (b) uplink throughput versus number of users with bandwidth of \num{1.4} GHz, and (c) uplink throughput versus bandwidth with $K = \num{64}$ users.}}
\label{[Simulation]ActiveRatio_Throughput]}
\end{figure*}
We analyze the uplink throughput performance of the proposed rainbow BF scheme alongside baseline schemes. For simplicity, we set the total number of time slots equal to the number of users, i.e., $L=K$. The antenna structures for the baseline schemes and the proposed rainbow BF are described in Fig. \ref{[Simulation]ArrayStructures} and corresponding BF strategies are as follows:
\begin{itemize}
\item \tcb{\textbf{Beam hopping}: For BH, we consider two cases: BH with beam-squint using PS-only arrays with fixed low noise amplifiers (LNAs) and tunable PSs, and BH without beam-squint using TTD-only arrays with fixed LNAs and tunable TTDs. Specifically, for BH with beam-squint, the beamformer at each time slot is
\begin{equation}
\mathbf{w}^{(m)}\big(\mathbf{0},\mathbf{\Phi}^{(\ell)}\big)=\mathbf{a}^{(\sf {c})}(u_\ell,v_\ell), \forall m, \ell,
\end{equation}
where $\mathbf{a}^{\sf (c)}(u,v)$ is the array response vector at center frequency $f_{\sf c}$. For BH without beam-squint, the beamformer is given by
\begin{equation}
\mathbf{w}^{(m)}\big(\mathbf{T}^{(\ell)},\mathbf{0}\big)=\mathbf{a}^{(m)}(u_\ell,v_\ell), \forall m, \ell.
\end{equation}
During each time slot, the beam with maximum gain is directed to a single user, and the beam direction sequentially changes over time to accommodate all users in the coverage area.}
\item \textbf{Beam sharing}: In beam sharing, tunable LNAs and PSs are employed for flexible multi-beam generation.
The beamformer is given by \cite{he2022physical}
\begin{equation}
\mathbf{w}^{(m)}\big(\mathbf{0},\mathbf{\Phi}^{(\ell)}\big)=\frac{\sqrt{N_{\sf rx}}\sum_{k=1}^K\mathbf{a}^{\sf (c)}(u_k,v_k)}{\Vert\sum_{k=1}^K\mathbf{a}^{\sf (c)}(u_k,v_k)\Vert}, \forall m, \ell,
\end{equation}
where the beamformer weights have variable amplitudes enabled by tunable LNAs.

\item \textbf{Rainbow BF}: Unlike BH and beam sharing, the proposed rainbow BF steers beams at different subcarrier frequencies towards distributed directions across the coverage area, independent of user location information. This design principle allows for a cost-effective implementation using fixed LNAs, PSs, and TTDs. The beamformer $\mathbf{w}^{(m)}(\mathbf{T}^{(\ell)},\mathbf{\Phi}^{(\ell)})$ remains time-invariant and is optimized using Algorithm 1 with desired mapping II.
\end{itemize} \tcb{The three schemes embody fundamentally different philosophies for achieving multi-user coverage. BH adopts a time-division approach, sequentially directing a single full-gain beam. Beam sharing simultaneously generates multiple beams toward specific user locations while accepting beam gain degradation. In contrast, rainbow BF directs near-full-gain beams in multiple directions with frequency-dependent BF, collectively covering the entire coverage area.}

Fig. \ref{[Simulation]ActiveRatio_Throughput]} presents the numerical performance of rainbow BF along with baseline schemes. Here, MaxCH represents the maximum average channel SNR-based subcarrier allocation with water-filling power allocation, where subcarrier $m$ is allocated to user $k^{\star(m)}$ such that $k^{\star(m)}=\underset{k}{\mathrm{argmax}}\,\gamma_k^{(m)}$. The active user ratio is measured as the proportion of users allocated at least one subcarrier out of the total $K$ users in a given time slot. As shown in Fig. \ref{[Simulation]ActiveRatio_Throughput]} (a), rainbow BF achieves an almost \num{100}\tcr{\%} active user ratio due to its spatially distributed beams across different subcarrier frequencies. In contrast, conventional BH is limited to serving a small area per time slot, leading to a significant drop in active user ratio as the number of users increases. \tcb{Notably, when $K=256$, rainbow BF outperforms conventional BH by a factor of \num{12.1}. This improvement allows users to utilize \num{12.1} times more total uplink power budget even under the identical per-user power constraints.} \tcb{Among the two BH schemes, BH with beam-squint slightly outperforms BH without beam-squint in terms of active user ratio because the spatially squinted frequency-dependent beams lead to marginally wider collective beam footprints across frequencies.} In the beam sharing baseline scheme, multiple beams are formed in each time slot. However, since all subcarriers share the same beam with only minor variations due to beam-squint effects, beam sharing can support more users than BH but fewer users than rainbow BF. The gap in active user ratios between JSPA and MaxCH stems from their distinct allocation strategies. Specifically, MaxCH assigns subcarriers to users with the highest average channel SNR (typically those nearest to the \tcb{beam center for each subcarrier}), whereas JSPA adaptively allocates subcarriers to users with maximum or near-maximal average channel SNR to optimize overall uplink power resources, thereby maximizing the total uplink throughput.

\tcb{Fig. \ref{[Simulation]ActiveRatio_Throughput]} (b) presents the uplink throughput versus the number of users $K$.} \tcb{The proposed rainbow BF outperforms BH when $K\geq 8$, benefiting from its nearly \num{100}\tcr{\%} active user ratio. For $K<8$, both BH and beam sharing outperform rainbow BF.} This is because rainbow BF spatially distributes frequency-dependent beams across a wide area regardless of user locations. At low user density, only a small portion of subcarriers' beams are directed toward users, while the remaining beams miss users or point to empty regions.
 In contrast, BH employs user location-aware BF to steer beams at all subcarrier frequencies toward targeted users. \tcbb{Remarkably, among the two BH schemes, BH with beam-squint slightly outperforms BH without beam-squint when the number of users is large. This is because the frequency-dependent beam spreading induced by beam-squint enlarges the effective beam footprint, allowing more users to fall within the beam coverage. In dense-user
scenarios, this benefit outweighs the beam gain degradation
caused by frequency-dependent beam misalignment}
\tcb{For $K\geq 8$, rainbow BF excels due to its higher active user ratio and greater total uplink transmission power.} Its performance improves with increasing user density, as the probability of users aligning with subcarrier's beam increases, leading to more efficient frequency resource utilization. \tcb{Notably, the performance gap between conventional BH and rainbow BF widens with the increasing number of users $K$, reaching a \num{2.8}-fold increase in uplink throughput compared to BH at $K=256$.} This \emph{scalability} with increasing user numbers makes rainbow BF particularly well-suitable for LEO SATCOM, where a single satellite must serve a massive number of users. In addition, beam sharing marginally outperforms conventional BH for $4\leq K \leq 32$ due to its higher active user ratio. However, its performance declines below BH for $K>32$ as increased simultaneous beams cause significant beam gain degradation.

 \tcb{Finally, Fig. \ref{[Simulation]ActiveRatio_Throughput]} (c) shows the uplink throughput versus bandwidth for a fixed number of users, $K = 64$. The numerical results show that the proposed rainbow BF delivers significantly higher uplink throughput compared to conventional benchmark schemes across all considered bandwidth ranges. Furthermore, the steep slope of the performance curve indicates that our approach has superior scalability with respect to bandwidth.}
 
 \tcb{In summary, Fig. \ref{[Simulation]ActiveRatio_Throughput]} demonstrates that the proposed rainbow BF can boost both the number of simultaneously served users within the coverage area and the uplink throughput, while providing scalability in terms of both the number of users and bandwidth.}\footnote{\tcb{Although we considered VSAT-type users with high antenna gain\tcr{s} of \num{43.2} dBi for simulation, the proposed rainbow BF framework is readily extensible to direct-to-cell (D2C) scenarios with handheld users.}}

\tcb{\begin{remark} \rm\textbf{(Implications of Numerical Results for LEO Satellite Constellation Design)}: Our performance evaluations show that rainbow BF achieves approximately 2.8 times uplink throughput and support 12.1 times more simultaneous users compared to conventional BH systems. This implies that fewer satellite are required to deliver equivalent performance, thereby simplifying system deployment and reducing overall resource requirements for LEO satellite constellation design.
\end{remark}}
\subsection{Subcarrier \& Power Allocation Performance}
\begin{figure}[t]
\centering
\includegraphics[width=0.85\linewidth]{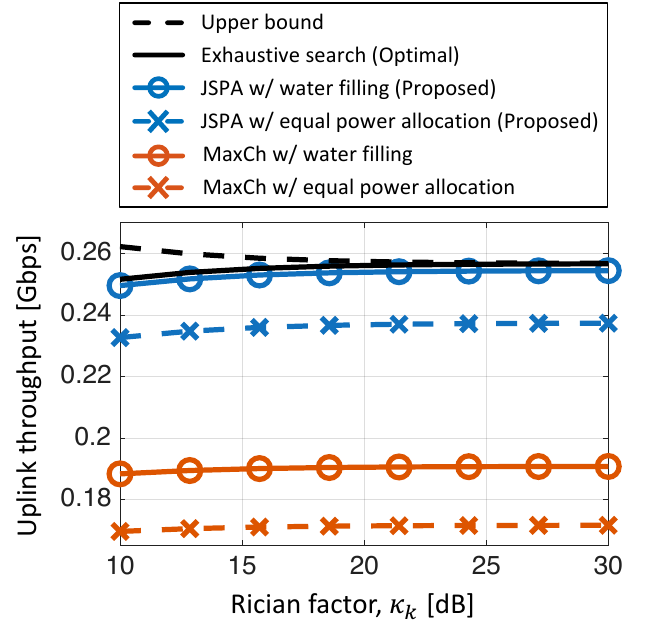}
\caption{\tcb{Uplink throughput comparison across different Rician factors for various subcarrier and power allocation schemes compared with the upper bound and the optimal value from exhaustive search (\num{10000} channel realizations).}}
\label{[Simulation]Resource_Allocation]}
\vspace{0mm}
\end{figure}
\tcb{Fig. \ref{[Simulation]Resource_Allocation]} presents the performance of the proposed JSPA algorithm under varying Rician factor $\kappa_k$ in the rainbow BF system. We evaluate six schemes: the proposed JSPA algorithm with water-filling, the proposed JSPA algorithm with equal power allocation, MaxCH with water-filling, MaxCH with equal power allocation, the upper bound, and the optimal scheme. The upper bound is obtained through an exhaustive search using perfect channel SNR information. Note that we assume perfect channel SNR information is not available due to short channel coherence time; hence, this provides a theoretical upper bound for evaluation of effectiveness of proposed schemes. The exhaustive search (optimal) performance is obtained by solving problem \eqref{Original Subcarrier/power allocation} using statistical channel SNR information via exhaustive search. We set $K=5$ and $M=8$ here for the exhaustive search computation. While our JSPA algorithm does not theoretically guarantee optimality, numerical results demonstrate that it achieves near-optimal performance. Furthermore, when the Rician factor $\kappa_k$ is sufficiently large, the performance of the proposed JSPA algorithm becomes closer to the upper bound even when relying on only partial CSI, demonstrating the effectiveness of the proposed approach. Notably, at $\kappa_k$ = 30 dB, the performance of the proposed JSPA with water-filling is only 0.87\tcr{\%} lower than the upper bound.}

\subsection{\tcb{Run Time Performance Evaluation}}

\begin{figure}[t]
\centering
\includegraphics[width=1\linewidth]{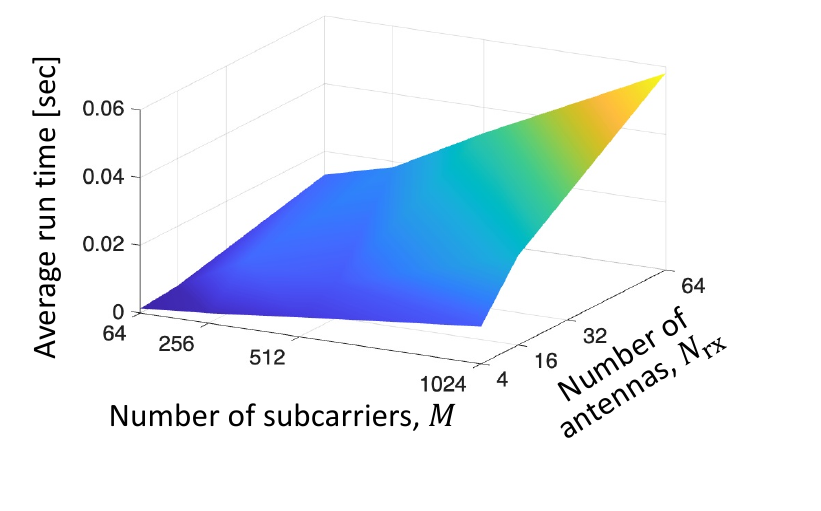}
\caption{\tcb{Average run time of the proposed rainbow BF algorithm with varying numbers of subcarriers and antenna elements.}}
\label{[Simulation]Algorithm1_Runtime}
\vspace{0mm}
\end{figure}

\tcb{To evaluate the practical feasibility of the proposed rainbow BF design algorithm, we analyze the runtime performance of Algorithm 1 with frequency-direction mapping II. 
Fig. \ref{[Simulation]Algorithm1_Runtime} presents the average run time for Algorithm 1 as a function of the number of subcarriers $M$ and antenna elements $N_{\sf rx}$. The simulations were conducted on MATLAB using a commercial laptop (Apple MacBook Pro with M4 Pro processor). We performed 1D line searches over \num{2000} grid points for each antenna element. When $M = \num{1024}$ subcarriers and $N_{\sf rx} = \num{64}$ antenna elements, the algorithm completes in 0.06 seconds. The sub-second computation times suggest real-world on-board processing feasibility, even on commercial processors without specialized hardware acceleration. Furthermore, the quasi-linear scaling behavior of computational complexity as a function of both the numbers of subcarriers and antenna elements establishes superior scalability.}

\section{Conclusion}
In this paper, we have proposed a novel rainbow BF approach that embraces beam-squint to fundamentally address the uplink throughput bottleneck in conventional BH-based LEO SATCOM. We have developed the alternating and decomposition-based optimization algorithm to achieve the desired 3D rainbow BF. \tcb{Extensive numerical results demonstrated that the proposed rainbow BF system achieves up to an \num{12.1}-fold increase in active user ratio and a \num{2.8}-fold improvement in uplink throughput compared to conventional BH.} We envision that this work will transform the paradigm in wideband LEO SATCOM, turning beam-squint from an avoidable limitation into a new degree-of-freedom for BF design and resource allocation, thereby inspiring future innovations in wideband 3D SATCOM. \tcb{Future research directions include developing adaptive frequency-direction mapping strategies that account for heterogeneous user distributions and varying traffic demands. Additionally, extending the proposed approach to multiple RF-chain architectures could further enhance system capacity. Finally, analyzing inter-satellite interference and developing mitigation techniques for multi-satellite constellation scenarios represents another promising avenue for investigation.}

\appendices
\section{Proof of Proposition 1}
\tcb{We prove the proposition by contradiction. Suppose there exists a solution $\{\mathbf{T}^\star,\boldsymbol{\Phi}^\star\}$ that satisfies Eq. \eqref{condition_Fullgain}. Let us select two distinct subcarriers $p$ and $q$ $(p\neq q)$ from the set of all subcarriers. From the assumption, the BF gains for both subcarriers can be maximized in their respective desired directions, i.e., 
\begin{align}
\mathbf{w}^{(r)}\left(\mathbf{T}^\star,\boldsymbol{\Phi}^\star\right) &= \alpha^{(r)}\mathbf{a}^{(r)}\left(u_{\sf des}^{(r)},v_{\sf des}^{(r)}\right),\,r\in\{p,q\}.
\label{Proposition 1 eq1}
\end{align}
$(N_x(n_x-1)+n_y)$-th element of vector equation \eqref{Proposition 1 eq1}, corresponding to the $(n_x,n_y)$-th antenna element, is given by
\begin{align}
&e^{j\left\{\phi^{\star(n_x,n_y)}-2\pi f_p\tau^{\star(n_x,n_y)}\right\}} 
\nonumber\\
&\qquad\qquad\qquad= \alpha^{(p)} e^{-j\pi\frac{f_p}{f_{\sf c}}\left\{(n_x-1)u_{\sf des}^{(p)}+(n_y-1)v_{\sf des}^{(p)}\right\}},
\label{[Proposition 1] eq1}
\\
&e^{j\left\{\phi^{\star(n_x,n_y)}-2\pi f_q\tau^{\star(n_x,n_y)}\right\}} 
\nonumber\\
&\qquad\qquad\qquad= \alpha^{(q)} e^{-j\pi\frac{f_q}{f_{\sf c}}\left\{(n_x-1)u_{\sf des}^{(q)}+(n_y-1)v_{\sf des}^{(q)}\right\}}.
\label{[Proposition 1] eq2}
\end{align}
Dividing both sides of Eq. (\ref{[Proposition 1] eq1}) by the corresponding sides of Eq. (\ref{[Proposition 1] eq2}) yields the following equation: 
\begin{align}
&e^{-j2\pi(f_p-f_q)\tau^{\star(n_x,n_y)}}\nonumber
\\
&=\frac{\alpha^{(p)}}{\alpha^{(q)}}e^{-j\frac{\pi}{f_{\sf c}}\left\{(n_x-1)\left(f_p u_{\sf des}^{(p)}-f_q u_{\sf des}^{(q)}\right)+(n_y-1)\left(f_p v_{\sf des}^{(p)}-f_q v_{\sf des}^{(q)}\right)\right\}}.
\end{align}
By assumption, this equation holds for all subcarriers $(p,q)$ and antenna elements $(n_x,n_y)$. Without loss of generality, let us consider two adjacent antenna elements in the $x$-direction: $(n_x+1,n_y)$ and $(n_x,n_y)$. Dividing their equations results in
\begin{align}
e^{-j2\pi(f_p-f_q)\left(\tau^{\star(n_x+1,n_y)}-\tau^{\star(n_x,n_y)}\right)}=e^{-j\frac{\pi}{f_{\sf c}}\left(f_p u_{\sf des}^{(p)}-f_q u_{\sf des}^{(q)}\right)}.
\label{[Proposition 1] eq3}
\end{align}
Then, considering the cyclic nature of complex exponential i.e., $e^{j\phi} = e^{j(\phi + 2\pi \beta)}$ for any integer $\beta$, we can re-express Eq. \eqref{[Proposition 1] eq3} as follows:
\begin{align}
&\tau^{\star(n_x+1,n_y)}-\tau^{\star(n_x,n_y)}=\frac{u_{\sf des}^{(p)}f_p-u_{\sf des}^{(q)}f_q }{2f_{\sf c}(f_p-f_q)}+\frac{\beta}{2\pi(f_p-f_q)}.\label{[Proposition 1] eq4}
\end{align}
where $\beta$ is arbitrary integer. By using the fact that $\frac{ax-by}{x-y}=a+\frac{a-b}{x-y}y$, Eq. (\ref{[Proposition 1] eq4}) can be re-expressed as follows:
\begin{align}
&\tau^{\star(n_x+1,n_y)}-\tau^{\star(n_x,n_y)}
\nonumber\\
&\qquad\qquad\quad
=u_{\sf des}^{(p)}+\frac{u_{\sf des}^{(p)}-u_{\sf des}^{(q)}}{2f_{\sf c}(f_p-f_q)}f_q+\frac{\beta}{2\pi(f_p-f_q)}.\label{[Proposition 1] eq5}
\end{align}
We select an additional distinct subcarrier $s$ $(p\neq q\neq s)$; without loss of generality, $u_{\sf des}^{(s)}$ can be expressed as
\begin{align}
    u_{\sf des}^{(s)}=u_{\sf des}^{(p)}-\frac{f_q}{f_s}\frac{f_p-f_s}{f_p-f_q}\left(u_{\sf des}^{(p)}-u_{\sf des}^{(q)}\right)-\frac{2f_{\sf c}(f_p-f_s)}{f_s}\epsilon,
\end{align} with arbitrary real number $\epsilon\in\mathbb{R}$. Eq. (\ref{[Proposition 1] eq5}) for subcarrier $p$ and $s$ is given by
\begin{align}
&\tau^{\star(n_x+1,n_y)}-\tau^{\star(n_x,n_y)}
\nonumber\\
&\qquad
=u_{\sf des}^{(p)}+\frac{u_{\sf des}^{(p)}-u_{\sf des}^{(q)}}{2f_{\sf c}(f_p-f_q)}f_q+\epsilon+\frac{\beta'}{2\pi(f_p-f_s)},\label{[Proposition 1] eq6}
\end{align}
where $\beta'$ is arbitrary integer. By assumption, Eqs. (\ref{[Proposition 1] eq5}) and (\ref{[Proposition 1] eq6}) should hold simultaneously. Subtracting them yields
\begin{align}
\epsilon\overset{(\rm{a})}{=}\frac{\beta}{2\pi(f_p-f_q)}-\frac{\beta'}{2\pi(f_p-f_s)},
\end{align}
where $\beta$ and $\beta'$ are arbitrary integers. Here, (a) holds when 
\begin{equation}
\epsilon\in\left\{\frac{\beta}{2\pi(f_p-f_q)}-\frac{\beta'}{2\pi(f_p-f_s)}~\Bigg|~\beta,\beta'\in\mathbb{Z}\right\}.
\label{[Proposition 1] eq7}
\end{equation}
The condition does not hold in general. To illustrate, consider a counterexample where $f_p=\frac{1.5}{2\pi}$, $f_q=\frac{1}{2\pi}$, and $f_s=\frac{0.5}{2\pi}$. In this case, Eq.~\eqref{[Proposition 1] eq7} yields $\epsilon\in\{2\beta-\beta'|\beta,\beta'\in\mathbb{Z}\}=\mathbb{Z}$, which contradicts that $\epsilon$ is an arbitrary real number. Thus, our initial assumption is invalid, completing the proof.}

\begingroup
\color{black}
\bibliographystyle{IEEEtran}
\bibliography{bibfile}
\endgroup

\end{document}